\documentclass[journal,twocolumn]{IEEEtran}
\usepackage{amsfonts,epsfig,array,multirow,ltablex,tabularx,setspace,arydshln,amssymb,multirow}
\usepackage{schemabloc,tikz}
\usepackage{cite}
\usepackage[cmex10]{amsmath}
\usepackage{graphicx}
\usepackage{amsmath}
\usepackage{amsfonts}
\usepackage[justification=centering]{caption}
\usepackage[font={small}]{caption}
\usepackage{epstopdf}
\usepackage{url}
\usepackage{array}
\usepackage{textcomp,bm}
\usepackage{soul,xcolor}

\usepackage{algorithm} 
\usepackage{algpseudocode}
\usepackage{pifont}
\usepackage{enumitem}
\usepackage{graphicx}
\usepackage{subcaption}
\usepackage{amsmath}
\usepackage{amsthm}
\usepackage[skip=2pt,font=small]{caption}
\usepackage{amssymb}

\newtheorem{lemma}{Lemma}
\newtheorem{definition}{Definition}
\newtheorem{theorem}{Theorem}
\usepackage{multirow}
\usepackage{comment}
\usepackage{resizegather}
\usepackage{makecell}

\ifCLASSINFOpdf
\else
\fi

\begin{document}
\setstcolor{red}
%
\title{M-RWTL: Learning Signal-Matched Rational Wavelet Transform in Lifting Framework}
\author{Naushad~Ansari,~\IEEEmembership{Student Member,~IEEE,}
        Anubha~Gupta,~\IEEEmembership{Senior Member,~IEEE}
\thanks{Naushad Ansari and Anubha Gupta are with SBILab, Deptt. of ECE, IIIT-Delhi, India (emails: naushada@iiitd.ac.in and anubha@iiitd.ac.in).
\newline Naushad Ansari is financially supported by CSIR (Council of Scientific and Industrial Research), Govt. of India, for this research work.}}
\maketitle

\begin{abstract}
Transform learning is being extensively applied in several applications because of its ability to adapt to a class of signals of interest. Often, a transform is learned using a large amount of training data, while only limited data may be available in many applications. Motivated with this, we propose wavelet transform learning in the lifting framework for a given signal. Significant contributions of this work are: 1) the existing theory of lifting framework of the dyadic wavelet is extended to more generic rational wavelet design, where dyadic is a special case and 2) the proposed work allows to learn rational wavelet transform from a given signal and does not require large training data. Since it is a signal-matched design, the proposed methodology is called Signal-\underbar{M}atched \underbar{R}ational \underbar{W}avelet \underbar{T}ransform Learning in the \underbar{L}ifting Framework (M-RWTL). The proposed M-RWTL method inherits all the advantages of lifting, i.e., the learned rational wavelet transform is always invertible, method is modular, and the corresponding M-RWTL system can also incorporate non-linear filters, if required. This may enhance the use of RWT in applications which is so far restricted. M-RWTL is observed to perform better compared to standard wavelet transforms in the applications of compressed sensing based signal reconstruction. 
\end{abstract}
\IEEEpeerreviewmaketitle
\section{Introduction}
\IEEEPARstart{T}{ransform} learning (TL) is currently an active research area and is being explored in several applications including image/video denoising, compressed sensing of magnetic resonance images (MRI), etc. \cite{ravishankar2013learning,wen2015video,ravishankar2015efficient}. Transform learning has the advantage that it adapts to a class of signals of interest and is often observed to perform better than existing sparsifying transforms such as total variation (TV), discrete cosine transform (DCT), and discrete wavelet transform (DWT) in the above said applications \cite{ravishankar2013learning,wen2015video,ravishankar2015efficient}. 

In general, transform learning is posed as an optimization problem satisfying some constraints that are specific to applications. Transform domain sparsity of signals is a widely used constraint along with some additional constraints on the transform to be learned, say, the minimization of Frobenius norm or the log-determinant of transform \cite{ravishankar2013learning}. The requirement of joint learning of both the transform basis and the transform domain signal under the constraints renders the optimization problem to be non-convex with no closed form solution. Thus, in general, TL problems are solved using greedy algorithms \cite{ravishankar2013learning}. A large number of variables (learned transform as well as transform coefficients) along with the greedy-based solution makes TL computationally expensive.

Recently, deep learning (DL) and convolutional neural network (CNN) based approaches are gaining momentum and are being used in several applications \cite{xie2012image,kulkarni2016reconnet}. In general, these methods (TL, DL, CNN) require a large amount of training data for learning. Hence, these methods may run with challenges in applications where only single snapshots of short-duration signals such as speech, music, or electrocardiogram (ECG) signals are available because a large amount of training data, required for learning transform, is absent. Hence, one uses existing transforms that are signal independent. This motivates us to look for a strategy to learn transform in such applications.

Among existing transforms, although Fourier transform and DCT find use in many applications, DWT provides an efficient representation for a variety of multi-dimensional signals \cite{mallat1999wavelet}. This efficient signal representation stems from the fact that DWT tends to capture signal information into a few significant coefficients. Owing to this advantage, wavelets have been applied successfully in many applications including compression, denoising, biomedical imaging, texture analysis, pattern recognition, etc. \cite{chang2000adaptive, gagie2012new, guido2006new, najarian2005biomedical, depeursinge2014rotation}. 

In addition, wavelet analysis provides an option to choose among existing basis or design new basis, thus motivates one to learn basis from a given signal of interest and/or in a particular application that may perform better than the fixed basis. Since the translates of the associated wavelet filters form the basis in $l^2$-space (functional space for square summable discrete-time sequences), wavelet transform learning implies learning wavelet filter coefficients. This reduces the number of parameters required to be learned with wavelet learning compared to the traditional transform learning. Also, the requirement of learning fewer coefficients allows one to learn basis from a short single snapshot of signal or from the small training data. This motivates us to explore wavelet transform learning for small data in this work. We also show that closed form solution exists for learning the wavelet transform leading to fast implementation of the proposed method without the need to look for any greedy solution. 

This is to note that most applications of DWT use dyadic wavelet transform, where the wavelet transform is implemented via a 2-channel filterbank with downsampling by two. In the frequency domain, this process is equivalent to decomposition of signal spectrum into two uniform frequency bands. $M$-band wavelet transform introduces more flexibility in analysis and decomposes a signal into $M$ uniform frequency bands. However, some applications, such as speech or audio signal processing, may require non-uniform frequency band decomposition \cite{blu1998new, choueiter2007implementation, bregovic2005design}.

Rational wavelet transform (RWT) can prove helpful in applications requiring non-uniform partitioning of the signal spectrum. Decimation factors of the corresponding rational filterbank (RFB) may be different in each subband and are rational numbers \cite{kovavcevic1993perfect}. RWT has been used in applications. For example, RWT is applied in context-independent phonetic classification in \cite{choueiter2007implementation}. In \cite{blanc1998using}, it is used for synthesizing 10m multispectral image by merging 10m SPOT panchromatic image and a 30m Landsat Thematic Mapper multispectral image. Wavelet shrinkage based denoising is presented in \cite{baussard2004rational} using signal independent rational filterbank designed by \cite{auscher1992wavelet}.  RWT is used in extracting features from images in \cite{le2014optimized} and in click frauds detection in \cite{chertov2010non}. Rational orthogonal wavelet transform is used to design optimum receiver based broadband Doppler compensation structure in \cite{yu2007optimum} and for broadband passive sonar detection in \cite{yu2011broadband}.

A number of RWT designs have been proposed in the literature including FIR (Finite Impulse Response) orthonormal rational filterbank design \cite{blu1998new}, overcomplete FIR RWT designs \cite{bayram2007design,bayram2009overcomplete}, IIR (Infinite Impulse Response) rational filterbank design \cite{bayram2009frequency}, biorthogonal FIR rational filterbank design \cite{nguyen2012design,nguyen2013bi}, frequency response masking technique based design of rational FIR filterbank \cite{bregovic2005design}, and complex rational filterbank design \cite{yu2006complex}. However, so far RWT designed and used in applications are meant to meet certain fixed requirements in the frequency domain or time-domain instead of learning the transform from a given signal of interest. For example, all the above designs are signal independent and hence, the concept of transform learning from a given signal has not been used so far in learning rational wavelets. 

Lifting has been shown to be a simple yet powerful tool for custom design/learning of wavelet \cite{sweldens1996lifting}. Apart from custom wavelet design/wavelet learning, lifting provides several other advantages such as a) wavelets can be designed in the spatial domain, b) designed wavelet transform is always invertible, c) the design is modular, and d) the design is DSP (Digital Signal Processing) hardware friendly from the implementation viewpoint \cite{sweldens1996lifting}. 
 However, the framework developed so far is used only for the custom design/learning of dyadic (or M-band) wavelets \cite{dong2008signal,blackburn2009two,kale2014lifting,kovavcevic2000wavelet} and has not been used to learn the rational wavelet transform to the best of our knowledge. Moreover, the existing architecture of lifting framework cannot be extended directly to rational filterbank structure because of different sample/signal rates in subband branches. 
 
From the above discussion, we note that the lifting framework can help in learning custom-design rational wavelets from given signals in a simple modular fashion that will also be easy to implement in the hardware.  Also, this may lead to the enhanced use of rational wavelet transforms in applications, similar to dyadic wavelets, which is so far restricted.

Motivated by the success of transform learning in applications, the flexibility of rational wavelet transform with respect to non-uniform signal spectral splitting, and the advantages of lifting in learning custom design wavelets, we propose to learn rational wavelet transform from a given signal using the lifting framework. The proposed methodology is called, ``Learning Signal-Matched Rational Wavelet Transform in the Lifting Framework (M-RWTL)". Following are the salient contributions/significance of this work:
\begin{enumerate}
\item The theory of lifting is extended from dyadic wavelets to rational wavelets, where dyadic wavelet is a special case. The concept of rate converters is introduced in predict and update stages to handle variable subband sample rates.  
\item Theory is proposed to learn rational wavelet transform from a given signal, where rational wavelets with any decimation ratio can be designed. 
\item FIR analysis and synthesis filters are learned that can be easily implemented in hardware. 
\item Closed form solution is presented for learning rational wavelet and thus, no greedy solution is required making M-RWTL computationally efficient. 
\item The proposed M-RWTL can be learned from a short snapshot of a single signal and hence, extends the use of transform learning from the requirement of large training data to small data snapshots. 
\item The utility of M-RWTL is demonstrated in the application of compressed sensing-based reconstruction and is observed to perform better than the existing dyadic wavelet transforms.   
\end{enumerate}

This paper is organized as follows. In section II, we briefly describe the theory of lifting corresponding to the dyadic wavelet system and the theory of rational wavelet. In Section III, we present the proposed theory of learning M-RWTL and some learned examples. Section IV presents simulation results in the application of compressive sensing based signal reconstruction. Some conclusions are presented in section VI. 

\textit{Notations}: Scalars are represented as lower case letters and, vectors and matrices are represented as bold lower case and bold upper case letters, respectively. 
\section{Brief Background}
In this section, we provide brief reviews on the theory of dyadic wavelet in lifting structure, rational wavelet system, and polyphase decomposition theory required for the explanation of the proposed work.
\subsection{Theory of Lifting in Dyadic Wavelet}
\label{Section for Theory of lifting}
A general dyadic wavelet structure is shown in Fig. \ref{Fig1}, where $G_l(z)$ and $G_h(z)$ are analysis lowpass and highpass filters and $F_l(z)$ and $F_h(z)$ are synthesis lowpass and highpass filters, respectively. Lifting is a technique for either factoring existing wavelet filters into a finite sequence of smaller filtering steps or constructing new customized wavelets from existing wavelets. In general, lifting structure has three stages: Split, Predict, and Update (Fig. \ref{Fig2}). 

\textit{Split}: This stage splits an input signal $x[n]$ into two disjoint sets of even $x_e[n]$ and odd indexed $x_o[n]$ samples (Fig. \ref{Fig2}). The original input signal $x[n]$ is recovered fully by interlacing this even and odd indexed sample stream. Wavelet transform associated with the corresponding filterbank is also called as the \textit{Lazy Wavelet} transform \cite{sweldens1996lifting}. A Lazy Wavelet transform can be obtained from the standard dyadic structure shown in Fig. \ref{Fig1} by choosing analysis filters as $G_l(z)=Z\{g_l[n]\}=1$, $G_h(z)=Z\{g_h[n]\}=z$ and synthesis filters as $F_l(z)=Z\{f_l[n]\}=1$, $F_h(z)=Z\{f_h[n]\}=z^{-1}$.   

\textit{Predict}: In this stage, one of the above two disjoint sets of samples is predicted from the other set of samples. For example, in Fig.2(a), odd samples are predicted from the neighboring even samples using the predictor filter $P \equiv T(z)$. Predicted samples are subtracted from the actual samples to calculate the prediction error. This step is equivalent to applying a highpass filter on the input signal. This stage modifies the analysis highpass and synthesis lowpass filters, without altering the other two filters, according to the following relations: 
\begin{equation}
G_{h}^{new}(z) = G_h(z)-G_l(z)T(z^{2}).
 \label{eq:no1}
\end{equation}
\begin{equation}
F_{l}^{new}(z) = F_l(z)+F_h(z)T(z^{2}).
 \label{eq:no2}
\end{equation}
\textit{Update}: In this stage, the other disjoint sample set is updated using the predicted signal obtained in the previous step via filtering through $U \equiv S(z)$ as shown in Fig.\ref{Fig2}. In general, this step is equivalent to applying a lowpass filter to the signal and modifies/updates the analysis lowpass filter and synthesis highpass filter according to the following relations:
\begin{equation}
G_{l}^{new}(z) = G_l(z)+G_h(z)S(z^{2}).
 \label{eq:no3}
\end{equation}
\begin{equation}
F_{h}^{new}(z) = F_h(z)-F_l(z)S(z^{2}).
 \label{eq:no4}
\end{equation}

One of the major advantages of lifting structure is that each stage (predict or update) is invertible. Hence, perfect reconstruction (PR) is guaranteed after every step. Also, the lifting structure is modular and non-linear filters can be incorporated with ease.
\begin{figure}[!ht]
\begin{center}
\includegraphics[scale=0.65, trim =6mm 2mm 6mm 8mm]{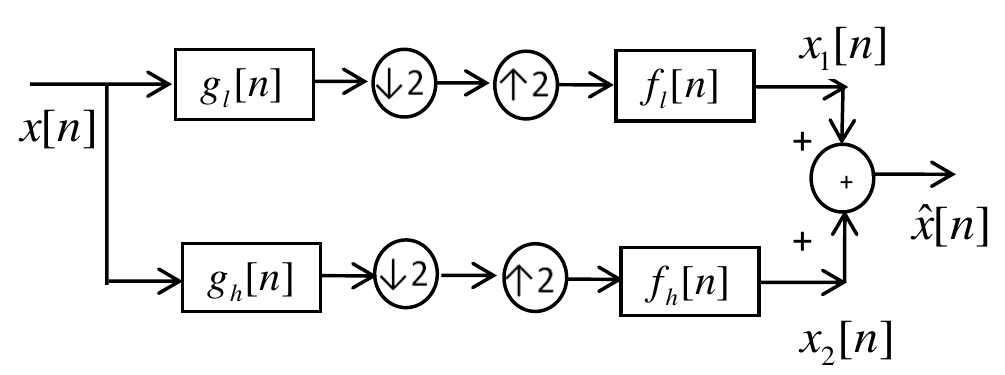}
\caption{Two Channel Biorthogonal Wavelet System}
\label{Fig1}
\end{center}
\end{figure}
\begin{figure}[!ht]
\begin{center}
\includegraphics[scale=0.65, trim =6mm 4mm 6mm 4mm]{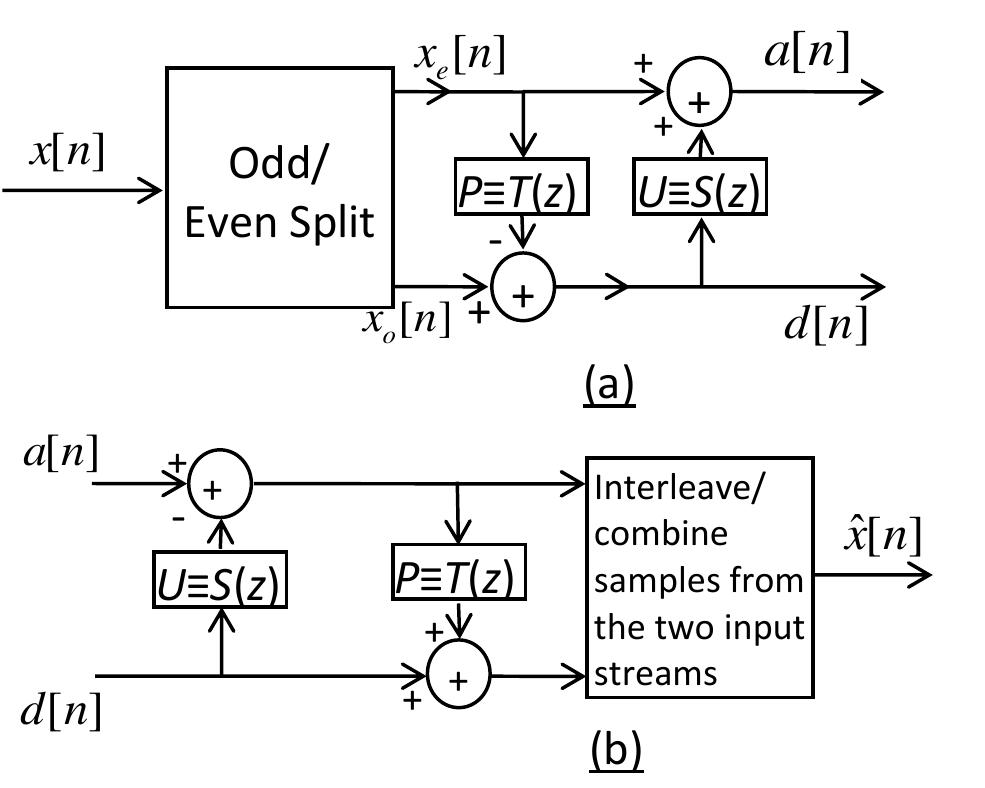}
\end{center}
\caption{Stages of Lifting: Split, Predict and Update}
\label{Fig2}
\end{figure}
\subsection{Rational Wavelet and Equivalent M-band structure} 
\label{Section for rational wavelet}
$M$-band wavelet system has integral downsampling ratio $M$ as shown in Fig. \ref{Fig for M-Band Wavelet}, whereas rational wavelet system has rational down-sampling ratios that allows decomposition of input signals into non-uniform frequency bands. In general, any $i^{th}$ analysis branch of a rational structure is as shown in Fig.\ref{Fig for General Rational Branch}, where $G_i(z)$ denotes the analysis filter, $q_i$ denotes the upsampling factor, and $p_i$ denotes the downsampling factor. For example, if $p_i=3$ and $q_i=2$, the downsampling ratio in this branch is equal to $3/2$. At the synthesis end, the order of downsampler and upsampler are reversed. 
\begin{figure}[!ht]
\centering
\includegraphics[scale=0.46]{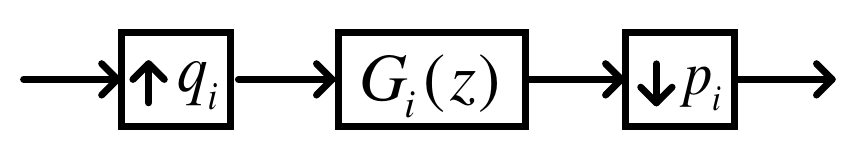}
\captionsetup{justification=centering}
\caption{$i^{th}$ branch of rationally decimated analysis filterbank}
\label{Fig for General Rational Branch}
\end{figure}
\begin{figure}[!ht]
\centering
\includegraphics[scale=0.46,trim =0mm 0mm 0mm 2mm]{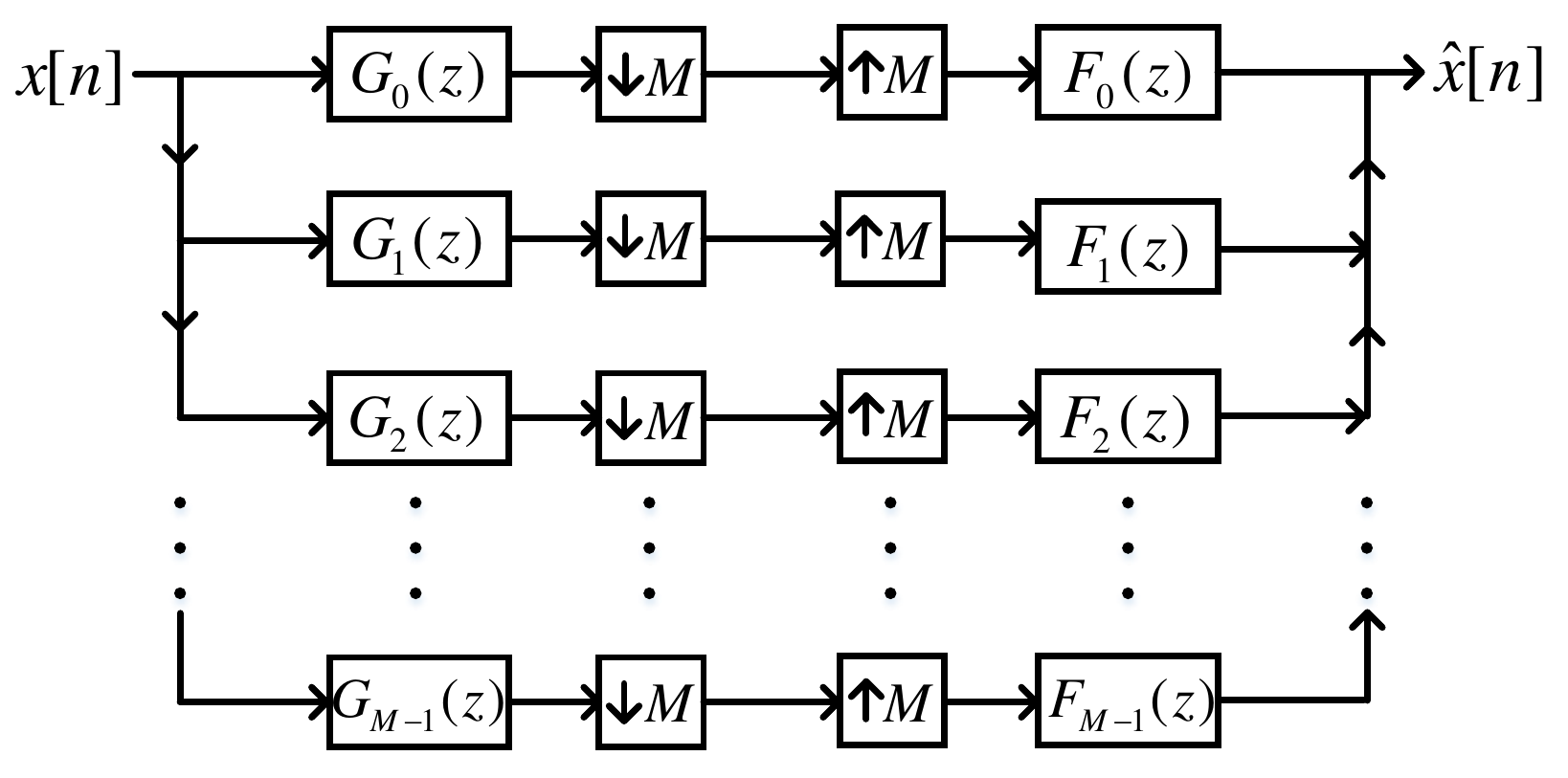}
\captionsetup{justification=centering}
\caption{\textit{M}-band wavelet structure}
\label{Fig for M-Band Wavelet}
\end{figure}
\begin{figure}[!ht]
\centering 
\captionsetup{justification=centering}
\includegraphics[scale=0.45,trim =0mm 0mm 0mm 0mm]{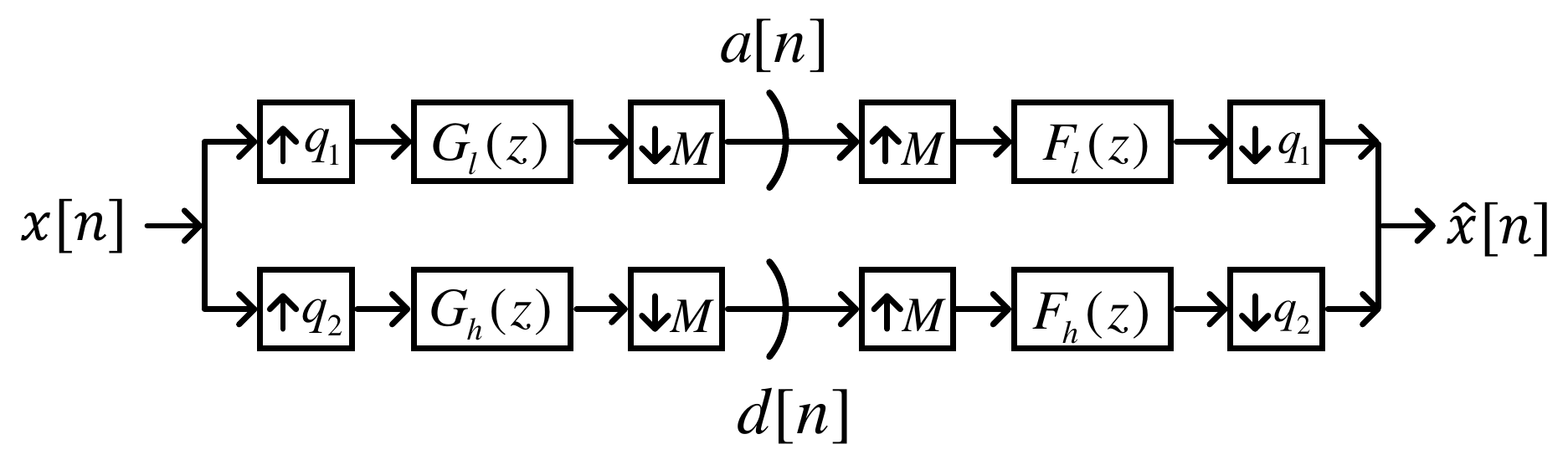}
\caption{General 2-band rational wavelet structure}
\label{Fig for General Rational Wavelet}
\end{figure}

An \textit{M}-channel rational filterbank is said to be critically sampled if the following relation is satisfied
\begin{equation}
\sum\limits_{i=0}^{M-1} \frac{q_i}{p_i}=1.
\label{eq:no5}
\end{equation}
This is to note that, throughout this paper, we consider all $p_i$'s to be equal, especially, $p_i=M$ for all \textit{i}. 

In general, a given $M$-band wavelet system, as shown in Fig.\ref{Fig for M-Band Wavelet}, can be converted into an equivalent $2$-band rational wavelet structure of Fig.\ref{Fig for General Rational Wavelet}, having downsampling ratios $\frac{M}{q_1}$ and $\frac{M}{q_2}$ in the two branches. Note that $q_1$ and $q_2$ are relatively prime with each other and with $M$. Also, $q_1+q_2=M$ for a critically sampled rational wavelet system. For example, the analysis filters $G_j(z)$ and synthesis filters $F_j(z)$ for $j={0,1,2,...,q_1-1}$ as shown in Fig.\ref{Fig for M-Band Wavelet} can be combined using the following equations:
\begin{equation}
G_l(z)=\sum_{i=0}^{q_1-1} z^{-iM}G_i(z^{q_1}),
\label{eq6}
\end{equation}
\begin{equation}
F_l(z)=\sum_{i=0}^{q_1-1} z^{iM}F_i(z^{q_1}),
\label{eq7}
\end{equation}
where $G_l(z)$ and $F_l(z)$ are the corresponding analysis and synthesis lowpass filters of the equivalent 2-band rational wavelet structure of Fig.\ref{Fig for General Rational Wavelet}. Similarly, rest of the filters of both sides can be combined using the following equations:
\begin{equation}
G_h(z)=\sum_{i=0}^{q_2-1} z^{-iM}G_{i+q_1}(z^{q_2}),
\label{eq8}
\end{equation}
\begin{equation}
F_h(z)=\sum_{i=0}^{q_2-1} z^{iM}F_{i+q_1}(z^{q_2}),
\label{eq9}
\end{equation}
where $G_h(z)$ and $F_h(z)$ are the corresponding higpass filters of analysis and synthesis ends of Fig.\ref{Fig for General Rational Wavelet}.
\subsection{Polyphase Representation and Perfect Reconstruction}
The polyphase representation of filters is very helpful in filterbank analysis and design \cite{vaidyanathan1993multirate}. Consider the \textit{M}-band critically sampled filterbank shown in Fig.\ref{Fig for M-Band Wavelet}. Analysis filter $G_i(z)$ can be written using type-1 polyphase representation as: 
\begin{equation}
G_i(z)=\sum_{j=0}^{M-1} z^{j}E_{i,j}(z^M),
\label{eq:no14}
\end{equation}
where $E_{i,j}(z)=g_i(j)+g_i(M+j)z^{M}+g_i(2M+j)z^{2M}+...$ and $G_i(z)=Z\{g_i[n]\}$.

Synthesis filter $F_i(z)$ can be written using type-2 polyphase representation as: \begin{equation}
F_i(z)=\sum_{i=0}^{M-1} z^{-j}R_{i,j}(z^M),
\label{eq:no15}
\end{equation}
where $R_{i,j}(z)=f_i(j)+f_i(M+j)z^{-M}+f_i(2M+j)z^{-2M}+...$ and $F_i(z)=Z\{f_i[n]\}$. The \textit{M}-band filterbank of Fig.\ref{Fig for M-Band Wavelet} can be equivalently drawn using polyphase matrices as shown in Fig.\ref{Fig for Polyphase rep}, where
\begin{equation}
E(z) = 
 \begin{pmatrix}
  E_{0,0} & E_{0,1} & \cdots & E_{0,M-1} \\
  E_{1,0} & E_{1,1} & \cdots & E_{1,M-1} \\
  \vdots  & \vdots  & \ddots & \vdots  \\
  E_{M-1,0} & E_{M-1,1} & \cdots & E_{M-1,M-1} 
 \end{pmatrix},
 \label{eq:no16}
\end{equation}

and 
\begin{equation}
R(z) = 
 \begin{pmatrix}
  R_{0,0} & R_{0,1} & \cdots & R_{0,M-1} \\
  R_{1,0} & R_{1,1} & \cdots & R_{1,M-1} \\
  \vdots  & \vdots  & \ddots & \vdots  \\
  R_{M-1,0} & R_{M-1,1} & \cdots & R_{M-1,M-1} 
 \end{pmatrix}.
 \label{eq:no17}
\end{equation}
\begin{figure}[!ht]
\centering
\includegraphics[scale=0.46]{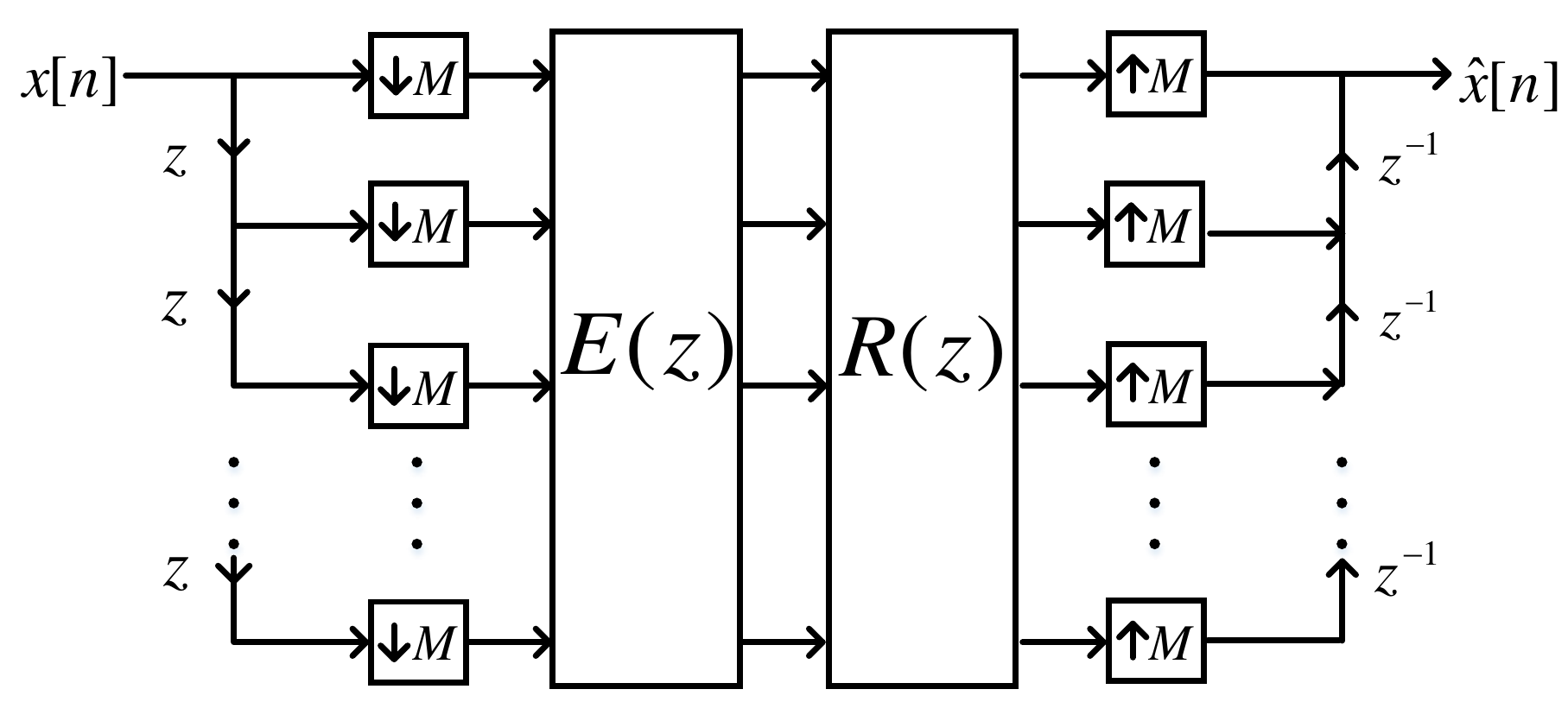}
\captionsetup{justification=centering}
\caption{\textit{M}-band wavelet structure with PR}
\label{Fig for Polyphase rep}
\end{figure}
The below relation of polyphase matrices 
\begin{equation}
 \mathbf{R}(z)\mathbf{E}(z)=cz^{-n_0}\mathbf{I},
  \label{PR equation}
\end{equation}
yields the condition of perfect reconstruction (PR) stated as
\begin{equation}
 \hat{x}[n]=cx[n-M-n_0],
  \label{Equation for PR condition}
\end{equation}
where $c$ is a constant and $n_0$ is a constant delay. 
  
\section{Proposed M-RWTL Learning Method for Signal-Matched Rational Wavelet}
\label{Section for proposed method}
In this section, we present the proposed M-RWTL method of learning signal-matched 2-channel rational wavelet system using the lifting framework. First, we propose the extension of 2-channel dyadic Lazy wavelet transform to $M$-band Lazy wavelet and find its equivalent 2-channel rational Lazy filterbank. Both these structures will be used in the proposed work. 

\subsection{\textit{M}-band and Rational Lazy Wavelet System}
\label{Section for Lazy Wavelet}
As explained earlier, a 2-band \textit{Lazy} wavelet system divides an input signal into two disjoint signals. Similarly, on the analysis side, an $M$-band Lazy wavelet system divides an input signal $x[n]$ into $M$ disjoint sets of data samples, given by $v_i[n],\,\,\, i={0,1,2,...,M-1}$, where $v_i[n]=x[Mn+i]$. At the synthesis end, these \textit{M} disjoint sample sets are combined or interlaced to reconstruct the signal at the output. An \textit{M}-band Lazy wavelet can be designed with the following choice of analysis and synthesis filters in Fig.\ref{Fig for M-Band Wavelet}:
\begin{align}
G_i(z)=&z^{i}    \qquad   i={0,1,2,...,M-1},\\
F_i(z)=&z^{-i}   \qquad  i={0,1,2,...,M-1}.
\label{eq:no1011}
\end{align}

In order to obtain the corresponding 2-band rational Lazy wavelet system with dilation factor $\frac{M}{q_1}$ (Fig.\ref{Fig for General Rational Wavelet}) from the $M$-band Lazy wavelet, we use \eqref{eq6} and \eqref{eq7} to obtain lowpass analysis and synthesis filters as:
\begin{align}
G_l(z)=& \sum_{i=0}^{q_1-1} z^{-iM}z^{iq_1} = \sum_{i=0}^{q_1-1} z^{-iq_2},  \nonumber \\
F_l(z)=& \sum_{i=0}^{q_1-1} z^{iM}z^{-iq_1} = \sum_{i=0}^{q_1-1} z^{iq_2}.
\label{eq:no1213}
\end{align}
Similarly, we use \eqref{eq8} and \eqref{eq9} to obtain the corresponding highpass analysis and synthesis filters of rational Lazy wavelet as:
\begin{align}
G_h(z)=& \sum_{i=0}^{q_2-1} z^{-iM}z^{(i+q_1)q_2} = z^{q_1q_2}\sum_{i=0}^{q_2-1} z^{-iq_1}, \nonumber \\
F_h(z)=& \sum_{i=0}^{q_2-1} z^{iM}z^{-(i+q_1)q_2} = z^{-q_1q_2} \sum_{i=0}^{q_2-1} z^{iq_1}.
\label{eq:no1213}
\end{align}
\begin{figure*}
\centering
\includegraphics[scale=0.435, trim=10 0 0 0]{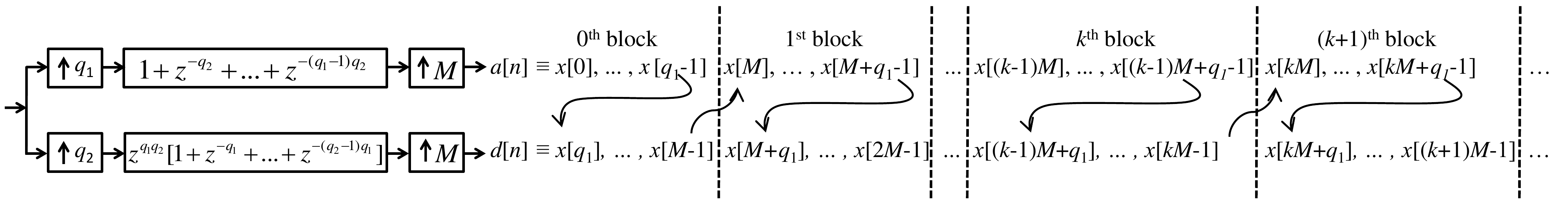}
\captionsetup{justification=centering}
\caption{Analysis side of rational Lazy wavelet; Each block consists of \textit{M} samples of input signal $x[n]$ divided into $a[n]$ and $d[n]$.}
\label{Lazy_Wavelet_output}
\end{figure*}
The above filters form the rational Lazy wavelet system equivalent of $M$-band Lazy wavelet transform. To learn signal-matched rational wavelet system, we start with the rational Lazy wavelet. This provides us initial filters $G_l(z)$, $G_h(z)$, $F_l(z)$, and $G_h(z)$. These filters are updated according to signal characteristics to obtain the signal-matched rational wavelet system. Analysis highpass and synthesis lowpass filters are updated in the predict stage, whereas analysis lowpass and synthesis highpass filters are updated in the update stage. These stages are described in the following subsections.

\subsection{Predict Stage}
\label{Section for Predict Stage}
As discussed in section \ref{Section for Theory of lifting}, a 2-band Lazy wavelet system divides the input signal into two disjoint sample sets, wherein one set is required to be predicted using the other set of samples. In the conventional 2-band lifting framework with integer downsampling ratio of \textit{M}=2 in both the branches (refer to Fig. 1 and 2), the output sample rate of $x_e[n]$ and $x_o[n]$ is equal. However, in a 2-band rational wavelet system, the output sample rate of two branches is unequal. Hence, the predict polynomial branch of a conventional 2-band dyadic design cannot be used. 

For example, from Fig.\ref{Fig for General Rational Wavelet}, we note that in a 2-band rational wavelet system, higher (lower) rate branch samples used in prediction should be downsampled (upsampled) by a factor of $k_p$ to equal the rate to the lower (higher) predicted branch samples, defined as:
\begin{equation}
k_p=\frac{q_1}{q_2},
\end{equation}
where $q_1/M$ is the rate of `\textit{predicting}' branch samples and $q_2/M$ is the rate of `to be \textit{predicted}' branch samples.

We propose to predict the lower branch samples with the help of upper branch samples.  Here, input signal $x[n]$ is divided into two disjoint sets, $x[kM+i], i=0,1,...,q_1-1$ and $x[kM+j], j=q_1,q_1+1,...,M-1$. We label these outputs as $\bm{a}$ and $\mathbf{d}$, respectively (Fig.\ref{Lazy_Wavelet_output}). Here, $k=0,1,...,L-1$, where $N=LM$ is the length of input signal $x[n]$ and, without loss of generality, is assumed to be a multiple of $M$. Thus, a given input signal is divided into \textit{L} blocks of $M$ size each at the subband output of 2-channel rational Lazy wavelet system. Here, first $q_1$ samples of every block move to the upper branch as a block of $\bm{a}$ and next $q_2$ samples (such that $q_1+q_2=M$) move to the lower branch as a block of $\mathbf{d}$. Or in other words, the rate of upper branch output is $q_1$ samples per block and the rate of lower branch output is $q_2$ samples per block.  Fig.\ref{Lazy_Wavelet_output} shows these blocks of outputs $\bm{a}$ and $\mathbf{d}$ explicitly.

This motivates us to introduce the concept of \textit{rate converter} that equals the output sample rate of the upper \textit{predicting} branch to that of the lower \textit{predicted} branch to enable predict branch design. In other words, the output of upper branch $a[n]$ is upsampled by $q_2$ and downsampled by $q_1$ to match the rate of lower branch output $d[n]$. It should be noted that the above downsampler and upsampler cannot be connected consecutively. Since an upsampler introduces spectral images in the frequency domain, it is generally preceded by a filter, also known as interpolator or anti-imaging filter. On the other hand, a downsampler stretches the frequency spectrum of the signal, that's why it is followed by a filter known as anti-aliasing filter. Both these conditions imply that if a rate converter is required with downsampling factor of $k_p=\frac{q_1}{q_2}$, a polynomial $R_p(z)$ should be incorporated which acts as the anti-imaging filter for the upsampler and at the same time as the anti-aliasing filter for the downsampler. The structure of this polynomial is $R_p(z)=1+z^{-1}+...+z^{-(q_2-1)}$. The placement of this polynomial is shown in Fig. \ref{Fig for Predict Rate Converter}. We call the complete branch with rate converter $k_p$ and polynomial $R_p(z)$ as the predict rate converter.

\begin{definition}
\textit{Predict Rate Converter}: It is a combination of the polynomial $R_p(z)=1+z^{-1}+...+z^{-(q_2-1)}$ preceded by a $q_2$-fold upsampler and followed by a $q_1$-fold downsampler as shown in Fig.\ref{Fig for Predict Rate Converter}. 
\end{definition}
\begin{figure}[!ht]
\centering 
\captionsetup{justification=centering}
\includegraphics[scale=0.45,trim =0mm 0mm 0mm 4mm]{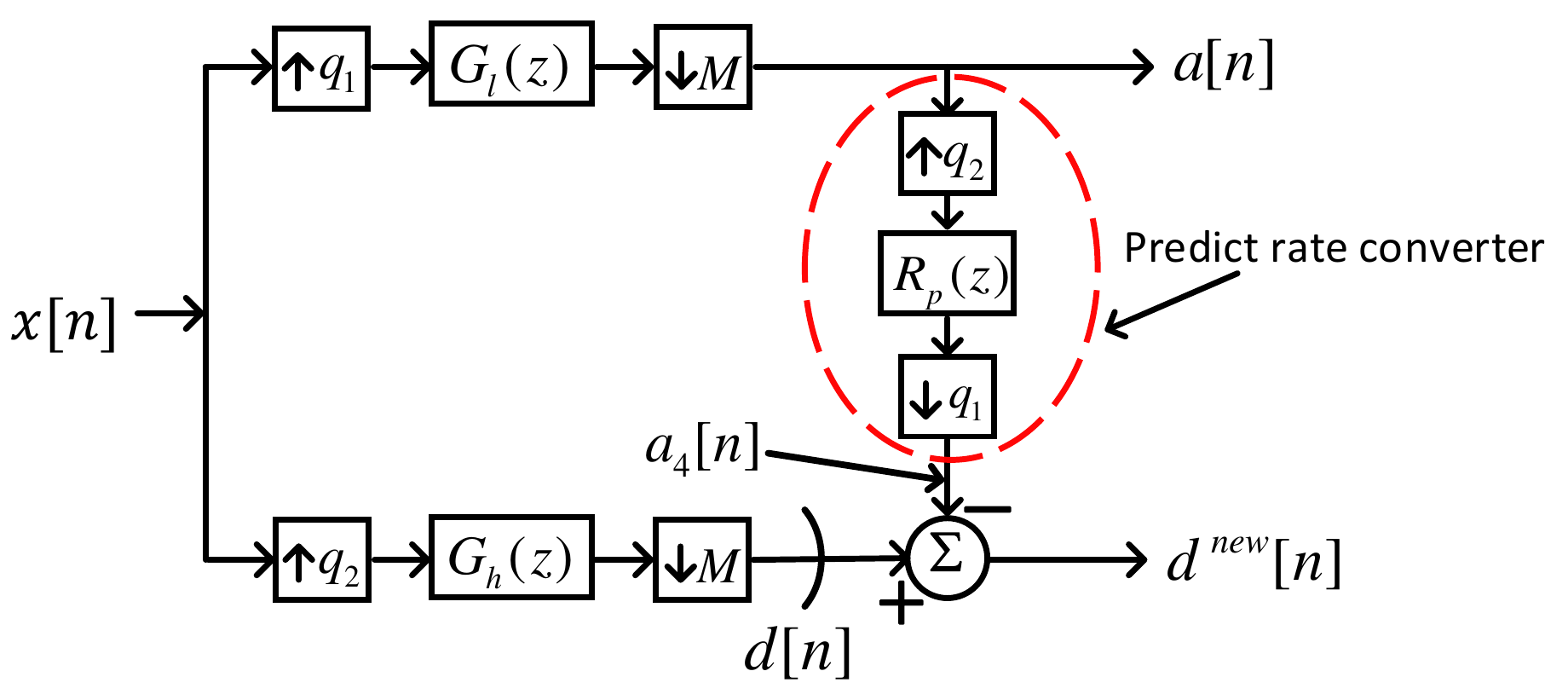}
\caption{Illustration of Predict Rate Converter}
\label{Fig for Predict Rate Converter}
\end{figure}

The predict rate converter with this choice of $R_p(z)$ will lead to an appropriate repetition or drop of samples of $a[n]$ such that the total number of samples in $a_4[n]$ and $d[n]$ in any $k^{th}$ block contains an equal number of samples.

Next, we present three subsections: 1) on the structure of predict stage filter $T(z)$, 2) how to learn $T(z)$ from a given signal, and 3) how to update all the filters of the corresponding 2-channel rational filterbank using the learned $T(z)$.\\

\subsubsection{Structure of predict stage filter $T(z)$}
Once the outputs of two analysis filterbank branches are equal, lower branch samples are predicted from the upper branch samples with the help of predict stage filter $T(z)$. This filter is introduced after the polynomial $R_p(z)$ as shown in Fig.\ref{Fig for General Predict Stage}. For good prediction, the current sample of input signal \textbf{x} contained in $d[n]$ should be predicted from its past and future samples. The $k^{th}$ block of $a[n]$ has preceding samples and $(k+1)^{th}$ block of $a[n]$ has future samples of the $k^{th}$ block of $d[n]$ as is evident from Fig.\ref{Lazy_Wavelet_output}. Thus, the structure of the predict polynomial $T(z)$ should be chosen appropriately. This is presented with a theorem for 2-tap $T(z)$ as below.

\begin{figure}
\centering 
\captionsetup{justification=centering}
\includegraphics[scale=0.45,trim =0mm 0mm 0mm 4mm]{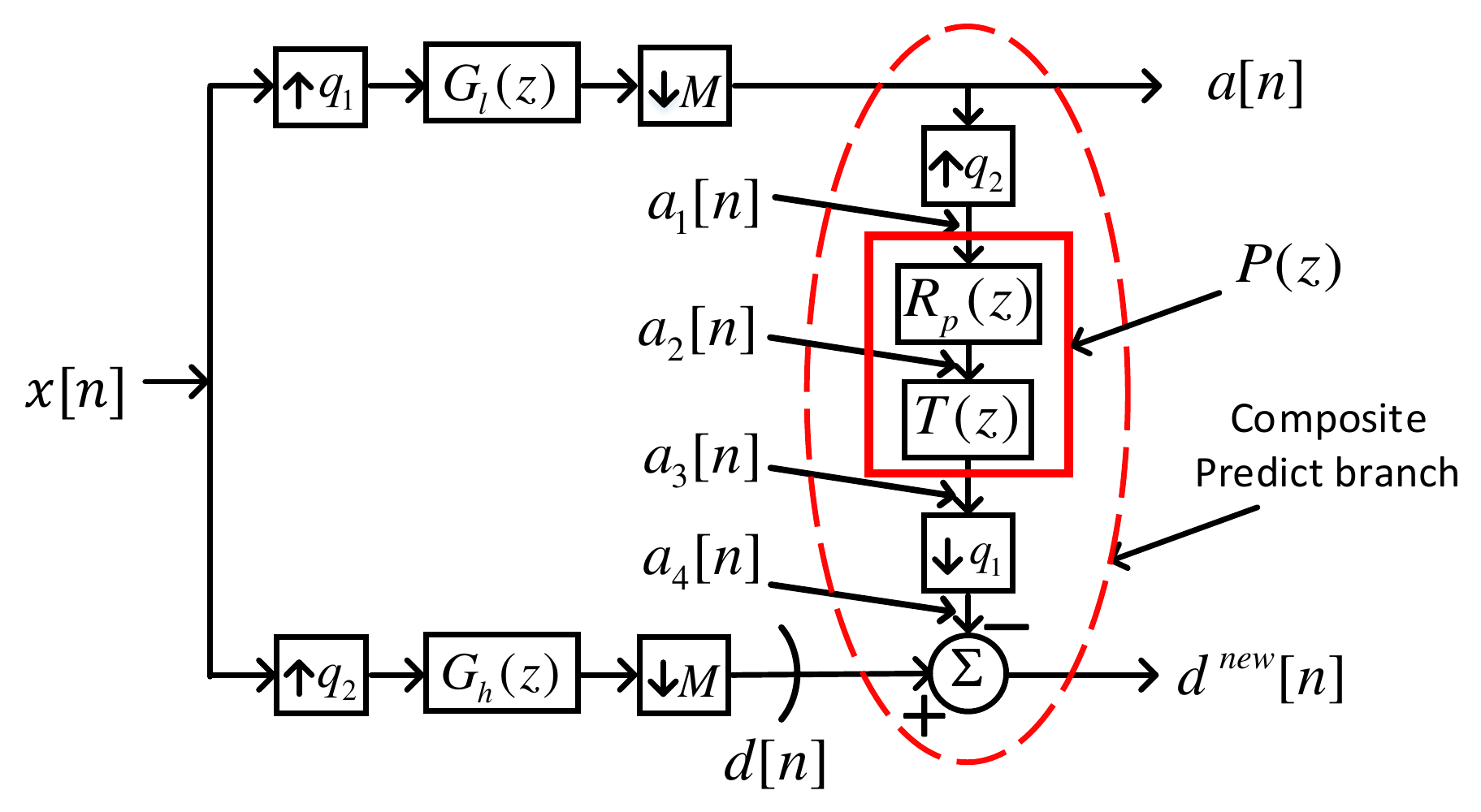}
\caption{Predict Rate Converter}
\label{Fig for General Predict Stage}
\end{figure}

\begin{theorem}
A 2-tap predict stage filter $T(z)=z^{q_2q_1}(t[0]z^{-q_2}+t[1])$ ensures that every sample in the $k^{th}$ block of $d[n]$ branch is predicted from the $k^{th}$ and/or the $(k+1)^{th}$ block samples of $a[n]$.
\end{theorem}
\begin{proof}
As discussed in section \ref{Section for Predict Stage}, on passing the input signal $x[n]$ through the Lazy filterbank of analysis side, we obtain approximate and detail coefficients ($\mathbf{a}$ and $\mathbf{d}$, respectively) in the form of blocks as shown in Fig.\ref{Lazy_Wavelet_output}. The $k^{th}$ block of approximate and detail signals are given by:
\begin{align}
a^k[m]=&a[(k-1)q_1+m]=x[(k-1)M+m], \nonumber \\
&\,\,\,\,m=0,1,...,q_1-1,   \nonumber \\        
d^k[m]=&d[(k-1)q_2+m]=x[(k-1)M+q_1+m],\nonumber \\   
&\,\,\,m=0,1,...q_2-1.
\label{a_explain}
\end{align} 
For prediction, first the upper branch signal $a[n]$ is passed through $q_2$-fold upsampler and polynomial $R_p(z)$. This leads to signal $a_2[n]$, as shown in Fig.\ref{Fig for General Predict Stage}, that contains every element of $a[n]$ repeated $q_2$ number of times. Mathematically, the $k^{th}$ block of signal $a_2[n]$ is given by:
\begin{equation}
a_2^{k}[m]=a\bigg[(k-1)q_1+\left\lfloor{\dfrac{m}{q_2}}\right\rfloor\bigg],\,\,\,m=0,1,...,q_1q_2-1,
\end{equation} 
where $\lfloor{.}\rfloor$ denotes the floor function. Predict filter $T(z)$ is defined as a product of two polynomials $z^{q_2q_1}$ and $t[0]z^{-q_2}+~t[1]$ and is applied to this signal $a_2[n]$. Intuitively, the first polynomial $z^{q_2q_1}$ will position the $(k+1)^{th}$ block of $a[n]$ over the $k^{th}$ block of $d[n]$ since every block of $a_2[n]$ contains $q_1q_2$ number of elements. The second polynomial $(t[0]z^{-q_2}+t[1])$ chooses two consecutive samples of $a[n]$ because every sample of $a[n]$ is repeated $q_2$ times as explained earlier. Thus, the second term will help in choosing either both the elements of $(k+1)^{th}$ block of $a[n]$ that are future samples of $d[n]$ or in choosing one element of $(k+1)^{th}$ block and one of $k^{th}$ block of $a[n]$. Mathematically, this is seen as below. 

On passing $a_2[n]$ through the predict filter $T(z)$, the $k^{th}$ block of signal $a_3[n]$ is obtained as: 
\begin{equation}
a_3^{k}[m]=t[0]a\bigg[(k-1)q_1+\left\lfloor{\dfrac{m}{q_2}}\right\rfloor+(q_1-1)\bigg]+t[1]a\bigg[(k-1)q_1+\left\lfloor{\dfrac{m}{q_2}}\right\rfloor+q_1\bigg],
\end{equation}
where $m=0,1,...,q_1q_2-1$. This signal is passed through the $q_1$-fold downsampler resulting in the $k^{th}$ block of signal $a_4[n]$ as:
\begin{align}
a_4^{k}[m]=&t[0]a\bigg[(k-1)q_1+\left\lfloor{\dfrac{q_1m}{q_2}}\right\rfloor+(q_1-1)\bigg]\nonumber \\ 
&+t[1]a\bigg[(k-1)q_1+\left\lfloor{\dfrac{q_1m}{q_2}}\right\rfloor+q_1\bigg],\nonumber \\
=&t[0]a\bigg[kq_1-1+\left\lfloor{\dfrac{q_1m}{q_2}}\right\rfloor\bigg]+t[1]a\bigg[kq_1+\left\lfloor{\dfrac{q_1m}{q_2}}\right\rfloor\bigg],
\label{a4_explain}
\end{align}
where $m=0,1,...,q_2-1$. The block size or the rate of $a_4[n]$ is same as that of $d[n]$ and hence, it predicts $d[n]$ providing the $k^{th}$ block prediction error given by:
\begin{equation}
d^{new,k}[m]=d^k[m]-a_4^k[m],\,\,\,m=0,1,...,q_2-1.
\label{dnew_explain}
\end{equation}
From \eqref{a_explain}, \eqref{a4_explain}, \eqref{dnew_explain}, and Fig.\ref{Lazy_Wavelet_output}, it can be noted that
\begin{enumerate}
\item[(i)] The first element of $k^{th}$ block of $d[n]$, i.e., $d^k[0]$ for $m=~0$ is $x[(k-1)M+q_1]$ and this sample is predicted from $a[kq_1-1]=x[(k-1)M+q_1-1]$ and $a[kq_1]=x[kM]$ that are the elements of the $k^{th}$ and $(k+1)^{th}$ blocks of $a[n]$, respectively. Also, these are the past and future samples of $x[(k-1)M+q_1]$.
\item[(ii)] The last element of $k^{th}$ block, i.e., $d[kq_2-1]=x[kM-1]$ for $m=q_2-1$ is predicted from $x[kM+q_1-2-\left\lfloor{\dfrac{q_1}{q_2}}\right\rfloor\bigg]$ and $x[kM+q_1-1-\left\lfloor{\dfrac{q_1}{q_2}}\right\rfloor\bigg]$ that are the elements of the $(k+1)^{th}$ block of $a[n]$.
\item[(iii)] From (i) and (ii) above, it is clear that in between elements of the $k^{th}$ block of $d[n]$ will be predicted from only the $k^{th}$ and $(k+1)^{th}$ block elements of $a[n]$. 
\end{enumerate}
In fact, $\lceil \frac{q_2}{q_1} \rceil$ elements of a block of $d[n]$ (where $\lceil{.}\rceil$ denotes the ceil function) are predicted using the past and future samples, i.e., from the $k^{th}$ and $(k+1)^{th}$ block elements of $a[n]$, while rest of the elements are predicted from the elements of $(k+1)^{th}$ block of $a[n]$. This proves Theorem-1.   
\end{proof}
We name the modified rate converter branch incorporating polynomial $T(z)$, shown in Fig.\ref{Fig for General Predict Stage}, as the `\textit{Composite Predict Branch}'.
 
This is to note that the above choice of polynomial provides the best possible generic solution for prediction using nearest neighbors for different values of $q_1$ and $q_2$. For example, if the \st{the} first polynomial $z^{q_2q_1}$ of $T(z)$ is omitted, one may note that samples in $d[n]$ will be predicted from far away past samples or far away future samples depending on the second polynomial. Thus, although $T(z)$ can be chosen in many ways, we choose to use the polynomial provided in Theorem-1.

The above theorem of 2-tap predict filter can be easily extended to obtain an $N_p$-tap filter with even $N_p$. For example, a 4-tap $T(z)$ can be given by
\begin{equation}
T(z)=z^{q_2q_1}(t[0]z^{-2q_2}+t[1]z^{-q_2}+t[2]+t[3]z^{q_2}),
\end{equation}
that will choose four consecutive samples of $a[n]$ and are from the immediate neighboring blocks of $d[n]$. In general, an even length $N_p$-tap filter $T(z)$ is given by
\begin{equation}
T(z)=z^{q_2q_1}z^{-\frac{N_p}{2}q_2}\sum_{k=0}^{N_p-1} t[k]z^{kq_2}.
\end{equation}

\subsubsection{Estimation of predict stage filter $T(z)$ from a given signal} In order to estimate an $N_p$-length predict stage filter $T(z)$, we consider the prediction error $d^{new}[n]$ shown in Fig.\ref{Fig for General Predict Stage} as below:
\begin{align}
d^{new}[n]=&d[n]-a_4[n]  \nonumber \\
= &d[n]-\sum_{j=0}^{N_p-1}\sum_{k=0}^{q_2-1}t[j]a_1[q_1n-(j+k)]),
\end{align}
and minimize it using the Least Squares (LS) criterion yielding the solution of $T(z)$ as below:
\begin{equation}
{\mathbf{\hat{t}}}=\min_{\mathbf{t}} \sum_{n=0}^{N_d-1} ({d^{new}[n]})^2,
\end{equation} 
where $N_d$ is the length of the difference signal $\mathbf{d}^{new}$ and $\mathbf{t}=[t[0] t[1] \dots t[N_p-1]]'$ is the column vector of elements of the polynomial $T(z)$. The above equation can be solved to learn the predict stage filter $T(z)$.

\subsubsection{Update of RFB filters using learned $T(z)$}
For the dyadic wavelet design using lifting as discussed in Section II.A, predict polynomial is used to update analysis highpass and synthesis lowpass filters using \eqref{eq:no1} and \eqref{eq:no2}. However, we require to derive similar equations for a rational filterbank. Before we present this work, let us look at a Lemma that will be helpful in defining these equations. 
\begin{lemma}
A structure containing a filter $H(z)$ followed by an $M$-fold downsampler and preceded by an $M$-fold upsampler (Fig. 10(a)) can be replaced by an equivalent filter $\tilde{H}(z)$ (Fig. 10(b)), which is given by: 
\begin{equation}
	\tilde{H}(z)=\frac{1}{M} \sum_{r=0}^{M-1} H(z^{\frac{1}{M}}W_M^r),
	\label{eq23}
\end{equation}
where $W_M=\exp(-j \frac{2 \pi}{M})$.
\end{lemma}
\begin{proof}
Refer section 4.3.5 of \cite{vaidyanathan1993multirate}.
\end{proof}
\begin{figure}[!ht]
\vspace{-2em}
\centering
\begin{subfigure}[b]{0.5\textwidth}
\centering
\captionsetup{justification=centering}
\includegraphics[scale=0.42, trim =0mm 0mm 0mm 0mm]{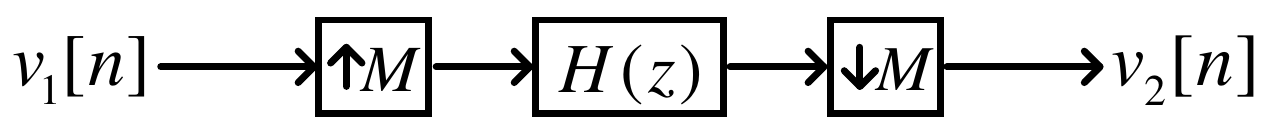}
\caption{}
\label{Fig for Lemma1 Part1}
\end{subfigure} 
\begin{subfigure}[b]{0.5\textwidth}
\centering
\captionsetup{justification=centering}
\includegraphics[scale=0.45]{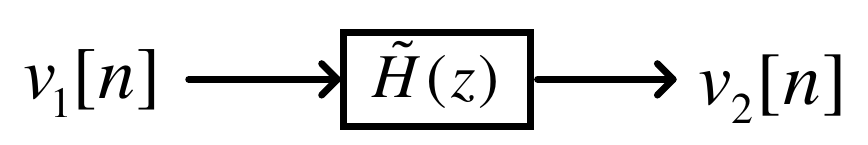}
\caption{}
\label{Fig for Lemma1 Part2}
\end{subfigure}
\captionsetup{justification=centering}
\caption{Filter structure in (a) is equivalent to filter in (b)}
\label{Fig for Lemma1}
\end{figure}
Next, we present Theorem-2 that provides the structure of updated analysis highpass filter $G_h^{new}(z)$ using $T(z)$ for an RFB.
\begin{theorem}
\label{th1_equation}
The analysis highpass filter of a 2-channel rational filterbank can be updated using the predict polynomial $T(z)$ used in the `composite predict branch' via the following relation:
\begin{equation}
G_h^{new}(z)=G_h(z)-\frac{1}{q_1} \sum_{k=0}^{q_1-1} G_l(z^{\frac{q_2}{q_1}}W_{q_1}^{q_2k})P(z^{\frac{M}{q_1}}W_{q_1}^{Mk}),
\label{eq_th1}
\end{equation}
\end{theorem}
where $P(z)=R_p(z)T(z)$. 
\begin{proof}
Refer to Fig.\ref{Figure for theorem-2}. Since $M$ and $q_2$ are relatively prime, the corresponding downsampler and upsampler can be swapped to simplify the structure in Fig.11(b). Using noble identities \cite{vaidyanathan1993multirate} and further simplification, we obtain the structure in Fig.11(c). As two downsampler or upsampler can swap each other, we obtain the structure in Fig.11(d). The part in red dotted rectangle in the figure can be replaced by an equivalent filter, $\tilde{H}_p(z)$ using \eqref{eq23} of Lemma-1 and is given by:
\begin{equation}
\tilde{H}_p(z)=\frac{1}{q_1} \sum_{r=0}^{q_1-1} G_l(z^{\frac{q_2}{q_1}}W_{q_1}^{q_2 r})P(z^{\frac{M}{q_1}}W_{q_1}^{M r}).
\end{equation}   

Considering the structure in Fig.11(d), signals $a_4[n]$, $d[n]$, and $d^{new}[n]$ can be written in $Z$-domain as:
\begin{equation}
A_4(Z)=\frac{1}{M} \sum_{k=0}^{M-1} X(z^{\frac{q_2}{M}}W_{M}^{q_2 k})\tilde{H}_p(z^{\frac{1}{M}}W_{M}^{k}),
\label{eq27}
\end{equation} 
\begin{equation}
D(Z)=\frac{1}{M} \sum_{k=0}^{M-1} X(z^{\frac{q_2}{M}}W_{M}^{q_2 k})G_h(z^{\frac{1}{M}}W_{M}^{k}),
\label{eq28}
\end{equation}
\begin{align}
D^{new}(z)& =  X_l(z)-X_u(z) \nonumber \\
=& \frac{1}{M} \sum_{k=0}^{M-1} X(z^{\frac{q_2}{M}}W_{M}^{q_2 k})\big[G_h(z^{\frac{1}{M}}W_{M}^{k}) - \tilde{H}_p(z^{\frac{1}{M}}W_{M}^{k})\big].
\label{eq24}
\end{align}
The above relation is equivalent to applying a new filter $G_h^{new}(z)$ to the lower branch of the rational wavelet system (Fig.11(e)) and is given by
\begin{align}
G_h^{new}(z)= & G_h(z)-\tilde{H}_p(z) \nonumber \\ 
= & G_h(z)-\frac{1}{q_1} \sum_{k=0}^{q_1-1} G_l(z^{\frac{q_2}{q_1}}W_{q_1}^{q_2k})P(z^{\frac{M}{q_1}}W_{q_1}^{Mk})
\end{align}
This proves Theorem-2.
\end{proof}

\begin{figure*}[!ht]
\centering 
\captionsetup{justification=centering}
\includegraphics[scale=0.38,trim =0mm 0mm 0mm 0mm]{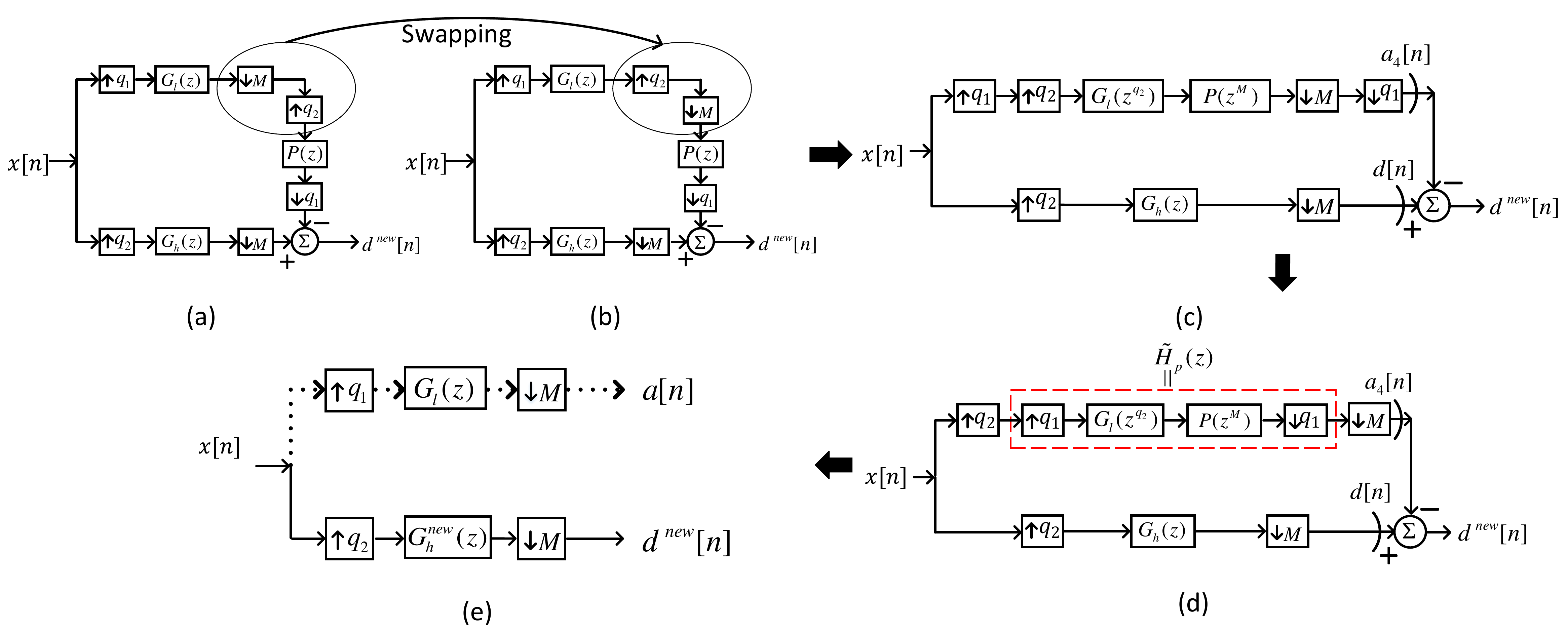}
\caption{Illustration for the proof of theorem-2. Structures in (a)-(e) are equivalent of each other.}
\label{Figure for theorem-2}
\end{figure*}

Next, the synthesis lowpass filter $F_l(z)$ is updated to $F_l^{new}(z)$ as follows. First, the analysis RFB containing $G_l(z)$ and $G_h^{new}(z)$ is converted into an equivalent $M$-band analysis filterbank structure as shown in Fig.\ref{Fig for M-Band Wavelet} using \eqref{eq6} and \eqref{eq8}, where lowpass filter $G_l(z)$ is transformed to $q_1$ upper filters of $M$-band analysis filters and highpass filter $G_h^{new}(z)$ of rational wavelet transforms to $M-q_1 (=q_2)$ lower filters of the $M$-band analysis filterbank. Next, the polyphase matrix $R^{new}(z)$ is obtained by using \eqref{eq:no14}, \eqref{eq:no16}, and \eqref{PR equation}. On using $R^{new}(z)$ from \eqref{eq:no17} in \eqref{eq:no15}, we obtain all $M$ updated synthesis filters of uniformly decimated filterbank. Out of these $M$ synthesis filters, lower $M-q_1$ filters are unchanged, while the upper $q_1$ filters are updated because of the predict branch. On using these filters in \eqref{eq7}, we obtain $F_l^{new}(z)$. 
\subsection{Update Stage}
\label{Section for Update Stage}
In this subsection, we present the structure of update branch to be used in lifting structure of an RFB, present the estimation of the update branch polynomial from the given signal, and the theorem for the update of corresponding filters of RFB. 

\subsubsection{Structure of update branch}
In the predict stage, we used upper branch samples to predict the lower branch samples. In the update stage, we update the upper branch samples using the lower branch samples. Since the output sample rate of two branches is unequal, we require to downsample the lower branch samples by a factor of $k_u$ given by:
\begin{equation}
k_u=\frac{q_2}{q_1},
\end{equation}
before adding this signal to $a[n]$ as shown in Fig.\ref{Figure for update stage}. 
\begin{figure}[!ht]
\centering 
\captionsetup{justification=centering}
\includegraphics[scale=0.38,trim =0mm 0mm 0mm 0mm]{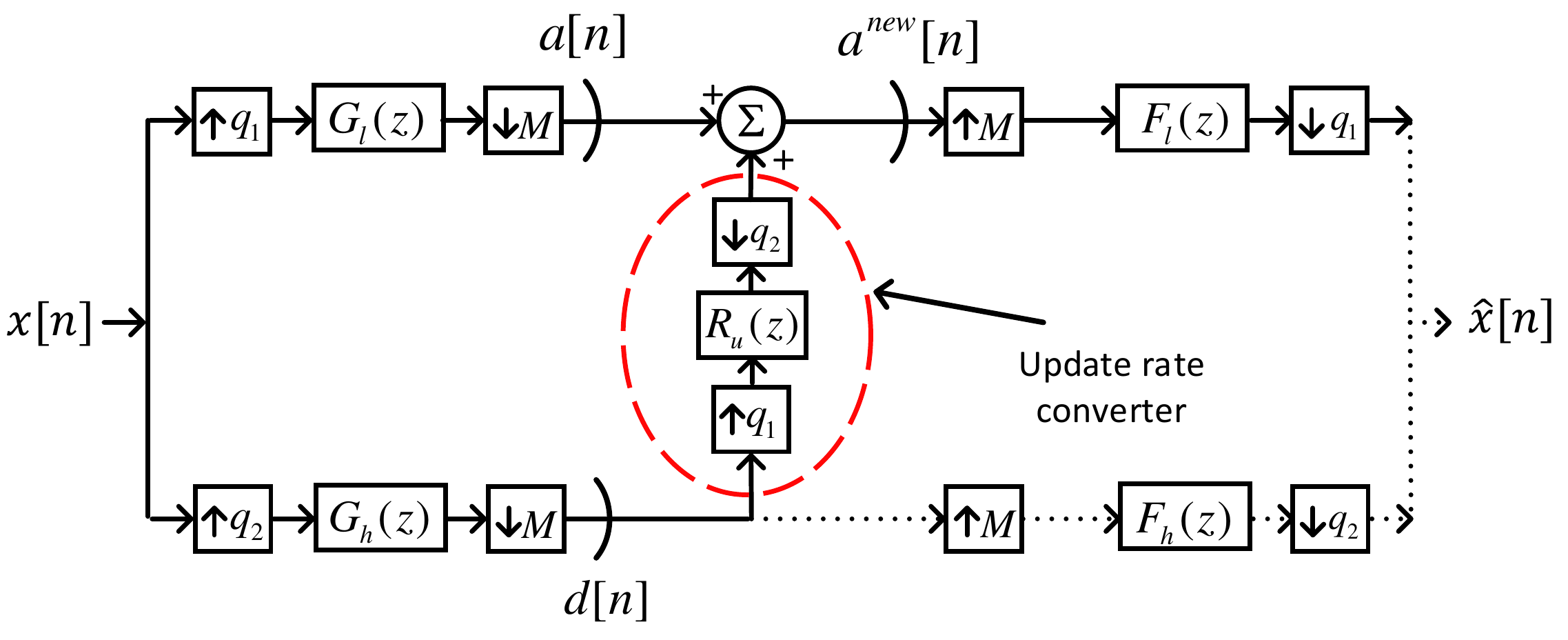}
\caption{Illustration of Update Rate Converter}
\label{Figure for update rate converter}
\end{figure}
\begin{figure}[!ht]
\centering 
\captionsetup{justification=centering}
\includegraphics[scale=0.38,trim =0mm 0mm 0mm 0mm]{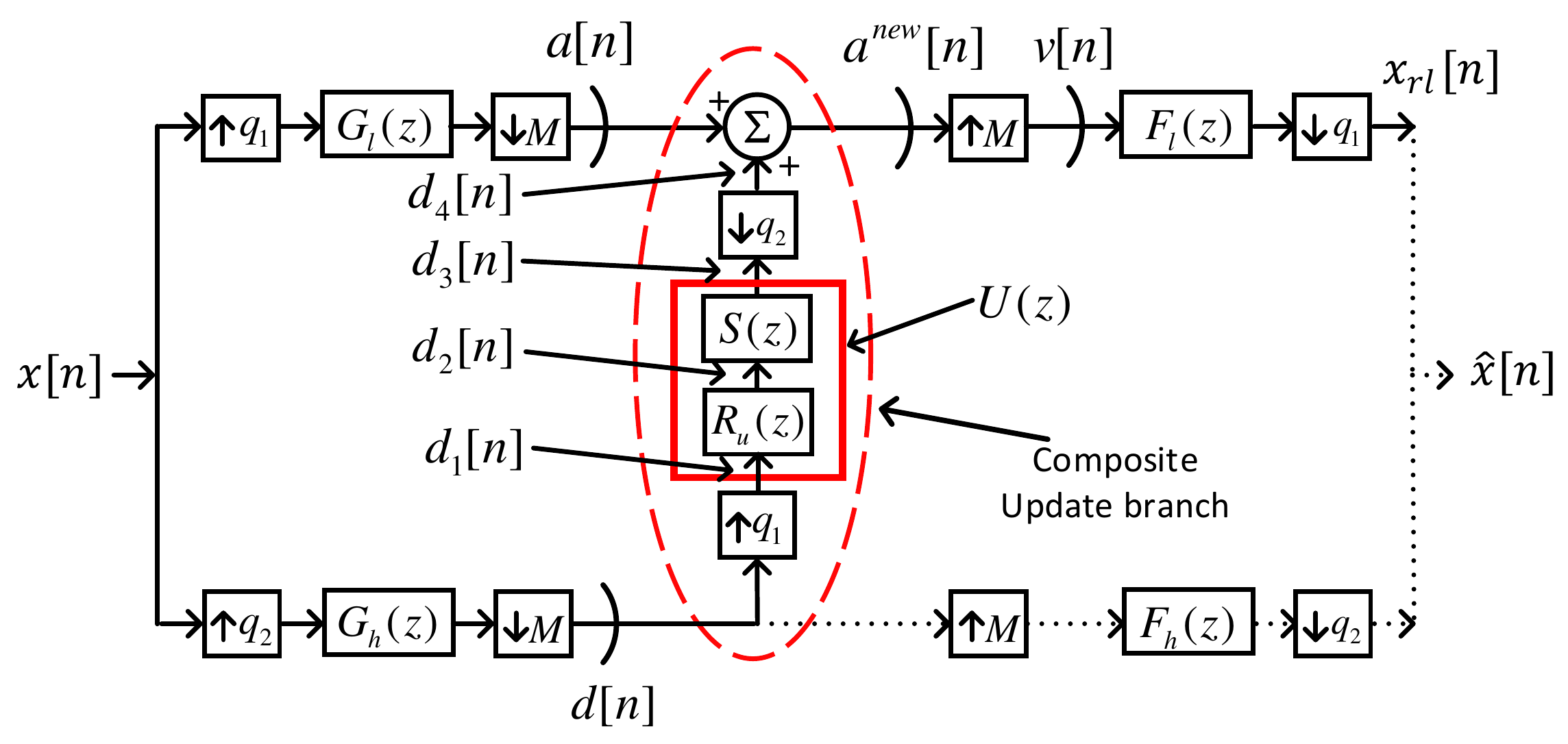}
\caption{Update Stage}
\label{Figure for update stage}
\vspace{-2em}
\end{figure}

As explained earlier, the upsampler is required to be followed by an anti-imaging filter and downsampler should be preceded by an anti-aliasing filter. We use filter $R_u(z)=1+z^{-1}+...+z^{-(q_1-1)}$ that accomplishes the same and its placement is shown in Fig.\ref{Figure for update rate converter}. Similar to the \textit{predict stage}, we call the update branch of Fig.\ref{Figure for update stage} as the \textit{composite update branch} that includes \textit{update rate converter}, update polynomial $S(z)$, and a summer.

\begin{definition}
Update rate converter: It is a combination of the polynomial $R_u(z)=1+z^{-1}+...+z^{-(q_1-1)}$ preceded by a $q_1$-fold upsampler and followed by a $q_2$-fold downsampler as shown in Fig.\ref{Figure for update rate converter}. 
\end{definition}

Similar to the structure of predict stage filter $T(z)$, the structure of update stage filter $S(z)$ should also be chosen carefully, so that the elements in the upper branch samples are updated only with the nearest neighbors. Below we present the theorem on the structure of a  2-tap update stage filter that ensures this.
\begin{theorem}
A 2-tap update stage filter $S(z)=s[0]+s[1]z^{-q_1}$ ensures that every sample in the $k^{th}$ block of $a[n]$  branch is updated from the $(k-1)^{th}$ and/or $k^{th}$ block samples of $d[n]$.
\end{theorem}
\begin{proof}
After passing the signal through the Lazy wavelet, blocks of output signal can be formed as described in section \ref{Section for Predict Stage}. The $k^{th}$ block of approximation and detail coefficients is given by
\begin{align}
a^k[m]=&a[(k-1)q_1+m]=x[(k-1)M+m], \nonumber \\
&\,\,\,\,m=0,1,...,q_1-1,   \nonumber \\        
d^k[m]=&d[(k-1)q_2+m]=x[(k-1)M+q_1+m],\nonumber \\   
&\,\,\,m=0,1,...q_2-1,
\end{align} 
where superscript denoted the $k^{th}$ block. For update, first the detail coefficients are passed through a $q_1$-fold upsampler and the polynomial $R_u(z)$. This leads to signal $d_2[n]$ that contain every element of $d[n]$ repeated $q_1$ times. Mathematically, $k^{th}$ block of signal $d_2[n]$ is given by:
\begin{equation}
d_2^{k}[m]=d\bigg[(k-1)q_2+\left\lfloor{\dfrac{m}{q_1}}\right\rfloor\bigg],\,\,\,m=0,1,...,q_1q_2-1,
\end{equation} 
Update filter $S(z)=s[0]+s[1]z^{-q_1}$ is applied on this signal $d_2[n]$. Unlike the predict stage, we do not require advancement of any block of $d_2[n]$ because $a_2[n]$ requires to be updated from the past and future samples of its block that are contained in the $(k-1)^{th}$ and the $k^{th}$ blocks of $d_2[n]$ (Fig. \ref{Lazy_Wavelet_output}). Hence, $S(z)$ requires only one polynomial $(s[0]+s[1]z^{-q_1})$ that chooses two consecutive samples of $d[n]$ because every sample of $d[n]$ is repeated $q_1$ times as explained earlier. 

On passing $d_2[n]$ through the update stage filter $S(z)$, we obtain
\begin{equation}
d_3^{k}[m]=s[0]d\bigg[(k-1)q_2+\left\lfloor{\dfrac{m}{q_1}}\right\rfloor\bigg]+s[1]d\bigg[(k-1)q_2+\left\lfloor{\dfrac{m-q_1}{q_1}}\right\rfloor\bigg],
\end{equation}   
where $m=0,1,...,q_1q_2-1$. This signal is downsampled by a factor of $q_2$ resulting in the $k^{th}$ block of $d_4[n]$ given by
\begin{equation}
d_4^{k}[m]=s[0]d\bigg[(k-1)q_2+\left\lfloor{\dfrac{q_2m}{q_1}}\right\rfloor\bigg]+s[1]d\bigg[(k-1)q_2+\left\lfloor{\dfrac{q_2m-q_1}{q_1}}\right\rfloor\bigg],
\label{update_eq1}
\end{equation}   
where $m=0,1,...,q_1-1$. Signal $d_4[n]$ helps with the update of signal in the upper branch. 

From \eqref{update_eq1}, we note that
\begin{enumerate}
\item[(i)] The first term of the $a^k[m]$, i.e., $a^k[0]=x[(k-1)M]$ for $m=0$ is updated with $d[(k-1)q_2]=x[(k-1)M+~q_1]$ and $d[(k-1)q_2-1]=x[(k-1)M-1]$ that are the elements of the $k^{th}$ and $(k-1)^{th}$ blocks of $d[n]$, respectively. Also, these are the past and future samples of $x[(k-1)M]$. 
\item[(ii)] The last term of $a^{k}[m]$, i.e., $a[kq_1-1]=x[(k-1)M+q_1-1]$ for $m=q_1-1$ is updated with $x\bigg[kM-\left\lfloor{\dfrac{q_2}{q_1}}\right\rfloor-1\bigg]$ and $x\bigg[kM-\left\lfloor{\dfrac{q_2}{q_1}}\right\rfloor-2\bigg]$ that are the elements in the $k^{th}$ block of detail signal.
\item[(iii)] From (i) and (ii) above, it is clear that in between elements of the $k^{th}$ block of $a[n]$ will be updated from only the $(k-1)^{th}$ and $k^{th}$ block elements of $d[n]$.
\end{enumerate}
This proves Theorem-3.   
\end{proof}

Similar to the predict stage polynomial, the above choice of polynomial provides the best possible generic solution for update using nearest neighbors for different values of $q_1$ and $q_2$. Thus, although $S(z)$ can be chosen in many ways, we choose to use the polynomial provided in Theorem-3.

The above theorem of 2-tap update filter can be easily extended to obtain an even length $N_s$-tap filter. For example, a 4-tap $S(z)$ can be given by
\begin{equation}
S(z)=s[0]z^{q_1}+s[1]+s[2]z^{-q_1}+s[3]z^{-2q_1},
\end{equation}
that will choose four consecutive samples of $d[n]$ and are from the immediate neighboring blocks of $a[n]$. In general, an even length $N_s$-tap filter $S(z)$ can be defined by the following relation:
\begin{equation}
S(z)=z^{(\frac{N_s}{2}-1)q_1}\sum_{k=0}^{N_s-1}s[k]z^{-kq_1},
\end{equation}
where $N_s$ is the length of the update stage filter, $S(z)$ and is assumed to be even.

\subsubsection{Learning of update stage filter $S(z)$ from a given signal} In order to estimate an $N_s$-length update stage filter $S(z)$, we consider the updated signal $a^{new}[n]$ shown in Fig.\ref{Figure for update stage} as below:
\begin{align}
a^{new}[n]=&a[n]+d_4[n] \nonumber \\
=&a[n]+\sum_{j=0}^{N_s-1} \sum_{k=0}^{q_1-1} s[j]d_1[q_2n-(j+k)]
\label{eq44}
\end{align}
On passing these approximate coefficients through an $M$-fold upsampler (Fig.\ref{Figure for update stage}), we obtain
\begin{align}
 v[n]=
	\begin{cases}
		a^{new}\left(\frac{n}{M}\right) & \quad \text{if } n \text{ is a multiple of $M$}\\
		0 & \quad \text{otherwise}.
	\end{cases}
	\label{eq45}
\end{align}
This signal $v[n]$ is passed through the synthesis lowpass filter $F_l(z)$ followed by a $q_1$-fold downsampler as shown in Fig.\ref{Figure for update stage} to obtain
\begin{equation}
x_{rl}[n]=\sum_{k=0}^{L_{f_l}} f_l[k]v[q_1n-k],
\label{eq46}
\end{equation}
where $\mathbf{x}_{{rl}}$ is the reconstructed signal at the lowpass branch of synthesis side with the same length as that of the input signal $\mathbf{x}$. $L_{f_l}$ is the length of the synthesis lowpass filter.

Assuming input signals to be rich in low frequency, most of the energy of the input signal moves to lowpass branch. Hence, signal $x_{rl}[n]$ is assumed to be the close approximation of the input signal $x[n]$. Correspondingly, the following optimization problem is solved to learn the update stage filter $S(z)$:
\begin{equation}
\mathbf{\hat{s}}=\min_{\mathbf{s}} \sum_{n=0}^{N-1} (x_{rl}[n]-x[n])^2,
\label{eq47}
\end{equation} 
where $\mathbf{s}=[s[0] s[1] \dots s[N_s-1]]'$ is the column vector of the coefficients of polynomial $S(z)$. From equation \eqref{eq44}, \eqref{eq45} and \eqref{eq46}, it can be observed that
signal $x_{rl}[n]$ can be written in terms of the update stage filter
$S(z)$. Thus, equation (\ref{eq47}) is solved using LS criterion to learn the update stage filter $S(z)$.
\begin{figure*}[!ht]
\centering 
\captionsetup{justification=centering}
\includegraphics[scale=0.32,trim =0mm 0mm 0mm 0mm]{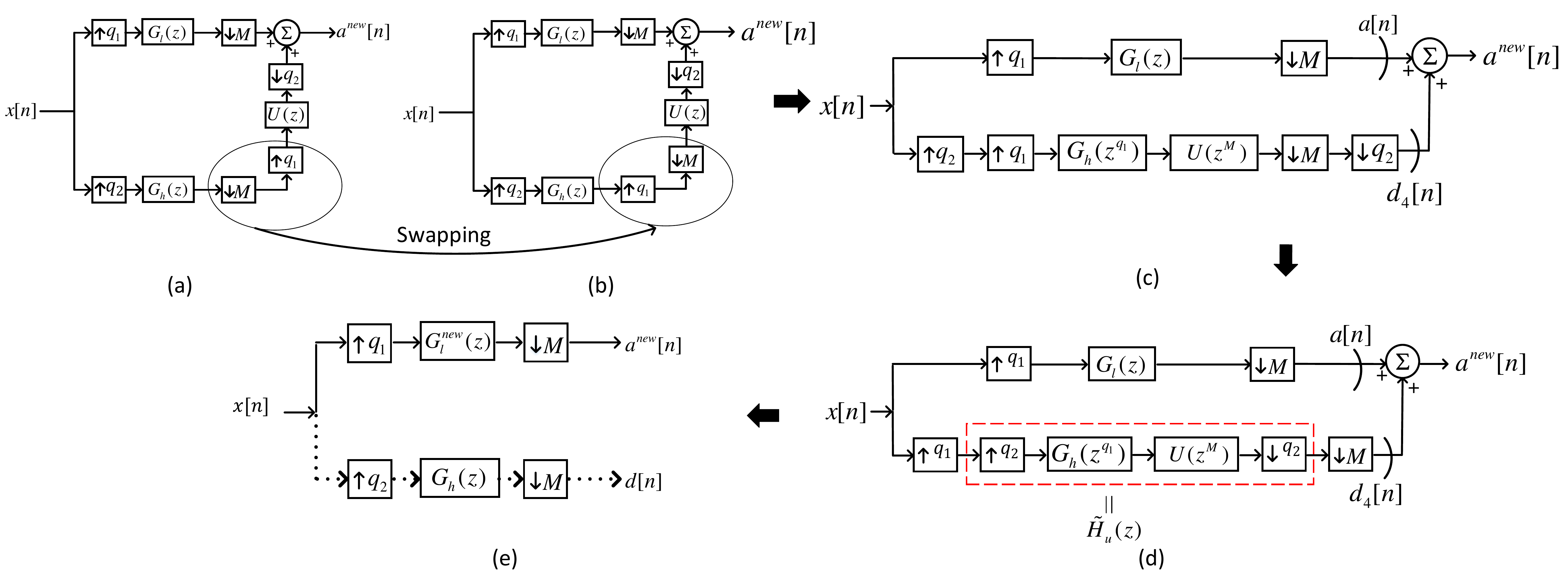}
\caption{Illustration for the proof of theorem-4. Structures in (a)-(e) are equivalent of each other}
\label{Figure for Theorem-4}
\end{figure*}

\subsubsection{Update of RFB filters using estimated $S(z)$}
Similar to the predict stage discussed earlier, we propose equation for the update of analysis lowpass filter using the update stage polynomial $S(z)$ for a rational filterbank in Theorem-4 below.
\begin{theorem}
The analysis lowpass filter of a 2-channel rational filterbank can be updated using the update polynomial $S(z)$ used in the `composite update branch' via the following relation:
\begin{equation}
\label{th2_equation}
G_l^{new}(z)=G_l(z)+\sum_{r=0}^{q_2-1} \frac{1}{q_2} G_h(z^\frac{q_1}{q_2}W_{q_2}^{q_1 r})U(z^\frac{M}{q_2}W_{q_2}^{M r}),
\end{equation}
where $U(z)=R_u(z)S(z)$.
\end{theorem}
\begin{proof}
Refer to Fig.\ref{Figure for Theorem-4}. Following the similar procedure as in Theorem 2, the update structure of Fig.14(a) can be equivalently converted to Fig.14(d) with the filter $\tilde{H}_u(z)$ given by:
\begin{equation}
\tilde{H}_u(z)=\frac{1}{q_2} \sum_{r=0}^{q_2-1} G_h(z^{\frac{q_1}{q_2}}W_{q_2}^{q_1 r}) U(z^{\frac{M}{q_2}}W_{q_2}^{M r}).
\end{equation}
Considering the structure in Fig.14(d), signals $d_4[n]$, $a[n]$, and $a^{new}[n]$ can be written in $Z$-domain as: 
\begin{equation}
D_4(z)=\frac{1}{M} \sum_{k=0}^{M-1} X(z^{\frac{q_1}{M}}W_{M}^{q_1 k})\tilde{H}_u(z^{\frac{1}{M}}W_{M}^{k}),
\label{eq39}
\end{equation} 
\begin{equation}
A(Z)=\frac{1}{M} \sum_{k=0}^{M-1} X(z^{\frac{q_1}{M}}W_{M}^{q_1 k})G_l(z^{\frac{1}{M}}W_{M}^{k}),
\label{eq40}
\end{equation}
\begin{align}
& A^{new}(z) = D_4(z)+A(z) \nonumber \\
& =\frac{1}{M} \sum_{k=0}^{M-1} X(z^{\frac{q_1}{M}}W_{M}^{q_1 k})\big[\tilde{H}_u(z^{\frac{1}{M}}W_{M}^{k}) + G_l(z^{\frac{1}{M}}W_{M}^{k})\big].
\label{eq24}
\end{align}
The above signal is equivalent to passing the signal $X(z)$ through an equivalent low pass filter $G_l^{new}(z)$ of the RFB (Fig.14(e)) and is given by:
\begin{align}
G_l^{new}(z) = & G_l(z)+\tilde{H}_u(z) \nonumber \\
= & G_l(z)+\frac{1}{q_2} \sum_{r=0}^{q_2-1} G_h(z^{\frac{q_1}{q_2}}W_{q_2}^{q_1 r}) U(z^{\frac{M}{q_2}}W_{q_2}^{Mr})
\end{align}
This proves the above theorem.
\end{proof}

Similar to the predict stage, the synthesis highpass filter $F_h^{new}(z)$ is learned using the updated analysis lowpass filter $G_l^{new}(z)$, polyphase matrices, and equations \eqref{eq6}, \eqref{eq8}, \eqref{eq9}, \eqref{eq:no14}-\eqref{PR equation}. Although we presented a specific case of this work with decimation ratio of $\big(\frac{2}{3},\frac{1}{3}\big)$ in \cite{ansari2015lifting}, the proposed method is very general and presents rational wavelet design theory for generalized rational factors. 

Note that the resultant filters of the rational filterbank are learned from the given signal because these are updated based on the predict and update stage filters $T(z)$ and $S(z)$ learned from the signal using (29) and (47), respectively. Also, the proposed method has closed-form solution because (29) and (47) can be directly solved using the least squares criterion.

\begin{figure*}
\begin{subfigure}[b]{0.24\linewidth}
\includegraphics[scale=0.43]{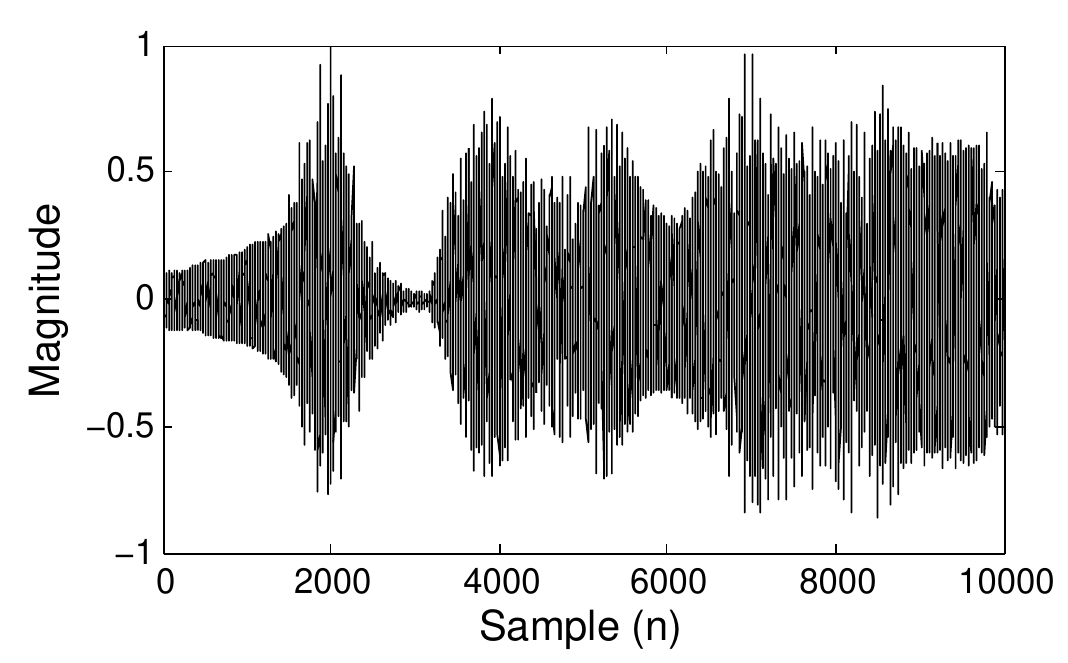}
\caption{Music signal-1}
\end{subfigure}
\begin{subfigure}[b]{0.24\linewidth}
\includegraphics[scale=0.43]{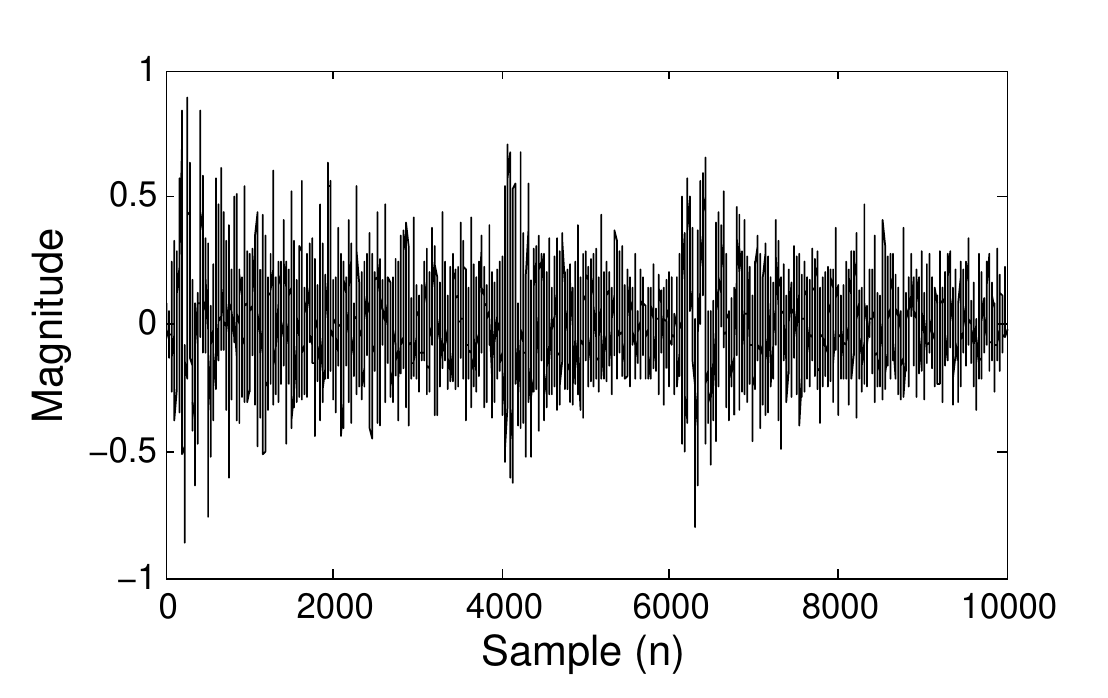}
\caption{Music signal-2}
\end{subfigure}
\begin{subfigure}[b]{0.24\linewidth}
\includegraphics[scale=0.43]{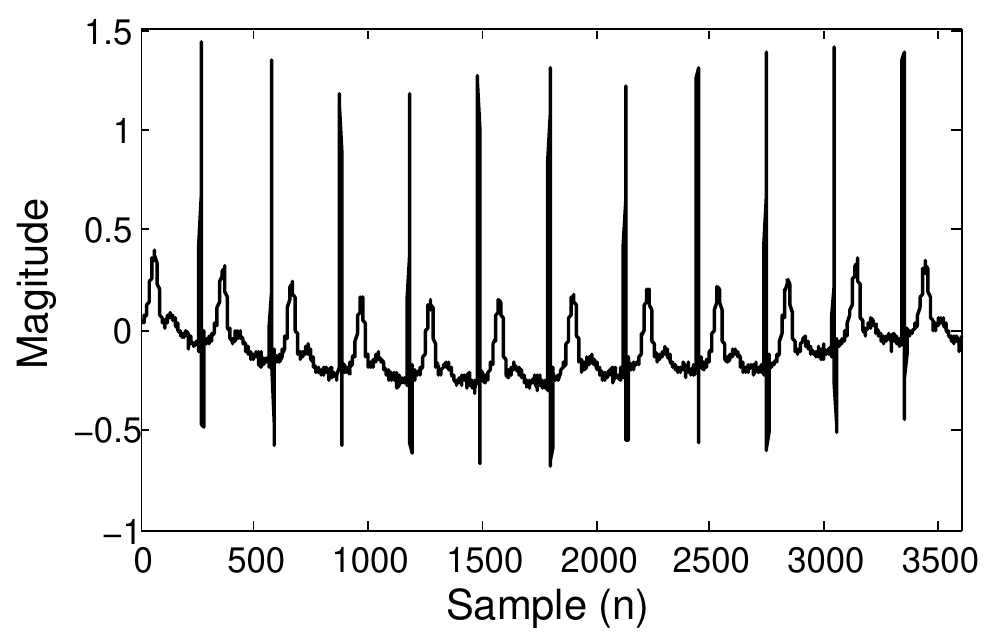}
\caption{ECG signal}
\end{subfigure}
\begin{subfigure}[b]{0.24\linewidth}
\includegraphics[scale=0.43,trim=0mm 0mm 0mm 0mm]{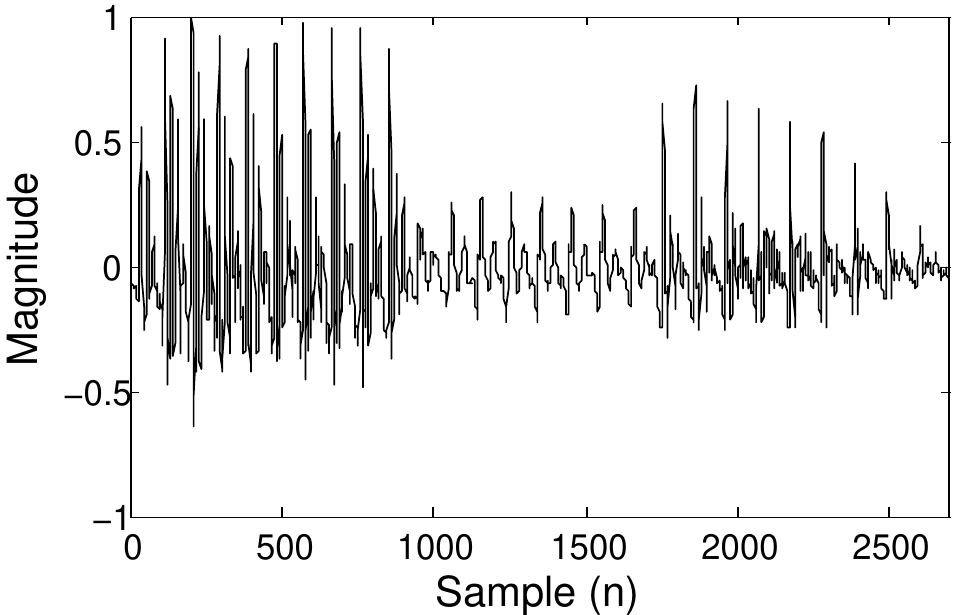}
\caption{Speech signal}
\end{subfigure}
\caption{Signals used in the experiments}
\label{Fig for all signals}
\end{figure*}
\begin{table*}[!ht]
\small
\centering
\caption{Illustration of M-RWTL learned based on the theory developed in section-\ref{Section for proposed method}}
\label{Table for design examples}
\begin{tabular}{cccccc}
\hline 
\makecell{Sampling rate \\ in two branches}  & $\big(\frac{2}{3},\frac{1}{3}\big)$ & $\big(\frac{1}{2},\frac{1}{2}\big)$ & $\big(\frac{3}{4},\frac{1}{4}\big)$ & $\big(\frac{2}{5},\frac{3}{5}\big)$ \\  \hline
($q_1,q_2,M$) &  $(2,1,3)$       &   $(1,1,2)$      &   $(3,1,4)$      &  $(2,3,5)$       \\ 
Lazy $G_l(z)$ & $1+z^{-1}$ & $1$              & $1+z^{-1}+z^{-2}$   & $1+z^{-3}$             \\ 
Lazy $G_h(z)$ & $z^2$      &   $z$            &         $z^3$               & $z^2+z^4+z^6$          \\
Lazy $F_l(z)$ & $1+z$      &   $1$            &         $1+z+z^2$           & $1+z^3$                \\
Lazy $F_h(z)$ & $z^{-2}$   & $z^{-1}$         &          $z^{-3}$            & $z^{-2}+z^{-4}+z^{-6}$ \\
$R_p(z)$             & $1$              &  $1$             &         $1$                  & $1+z^{-1}+z^{-2}$      \\
$T(z)$               & $z^2(t[0]z^{-1}+t[1])$  &$z(t[0]z^{-1}+t[1])$         & $z^3(t[0]z^{-1}+t[1])$   & $z^6(t[0]z^{-3}+t[1])$    \\
$R_u(z)$             & $1+z^{-1}$       & $1$              &         $1+z^{-1}+z^{-2}$    & $1+z^{-1}$             \\
$S(z)$               &$s[0]+s[1]z^{-2}$ &$s[0]+s[1]z^{-1}$ &         $s[0]+s[1]z^{-3}$    & $s[0]+s[1]z^{-2}$      \\ \hline     
\end{tabular}
\end{table*}

\begin{figure*}[!ht]
\centering
\begin{subfigure}[b]{0.24\textwidth}
\centering
\captionsetup{justification=centering}
\includegraphics[scale=0.35, trim =0mm 0mm 0mm 0mm]{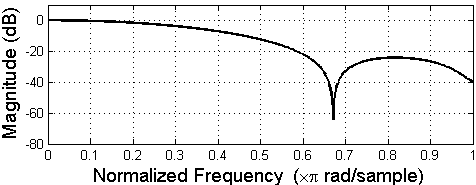}
\caption{$(\frac{2}{3},\frac{1}{3})$ LPF with Music-2}
\end{subfigure} 
\begin{subfigure}[b]{0.24\textwidth}
\centering
\captionsetup{justification=centering}
\includegraphics[scale=0.35]{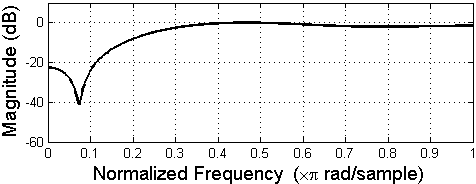}
\caption{$(\frac{2}{3},\frac{1}{3})$ HPF with Music-2}
\end{subfigure}
\begin{subfigure}[b]{0.24\textwidth}
\centering
\captionsetup{justification=centering}
\includegraphics[scale=0.35, trim =0mm 0mm 0mm 0mm]{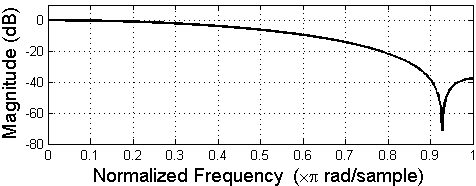}
\caption{$(\frac{1}{2},\frac{1}{2})$ LPF with ECG}
\end{subfigure} 
\begin{subfigure}[b]{0.24\textwidth}
\centering
\captionsetup{justification=centering}
\includegraphics[scale=0.35]{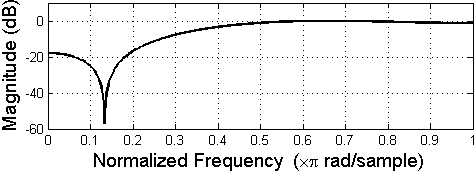}
\caption{$(\frac{1}{2},\frac{1}{2})$ HPF with ECG}
\end{subfigure}

\begin{subfigure}[b]{0.24\textwidth}
\centering
\captionsetup{justification=centering}
\includegraphics[scale=0.35, trim =0mm 0mm 0mm 0mm]{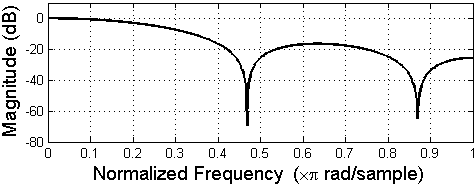}
\caption{$(\frac{3}{4},\frac{1}{4})$ LPF with Music-1}
\end{subfigure} 
\begin{subfigure}[b]{0.24\textwidth}
\centering
\captionsetup{justification=centering}
\includegraphics[scale=0.35]{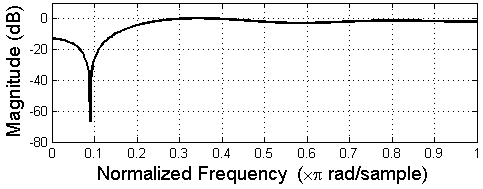}
\caption{$(\frac{3}{4},\frac{1}{4})$ HPF with Music-1}
\end{subfigure}
\begin{subfigure}[b]{0.24\textwidth}
\centering
\captionsetup{justification=centering}
\includegraphics[scale=0.35, trim =0mm 0mm 0mm 0mm]{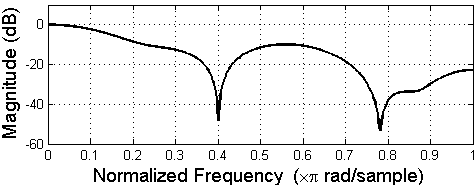}
\caption{$(\frac{2}{5},\frac{3}{5})$ LPF with Speech}
\end{subfigure} 
\begin{subfigure}[b]{0.24\textwidth}
\centering
\captionsetup{justification=centering}
\includegraphics[scale=0.35]{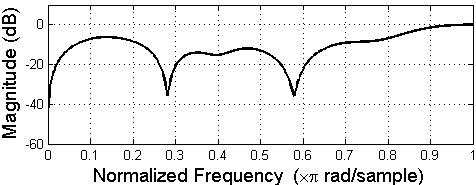}
\caption{$(\frac{2}{5},\frac{2}{5})$ HPF with Speech}
\end{subfigure}
\captionsetup{justification=centering}
\caption{Frequency response of synthesis filters presented in Table-\ref{Table for Filter Coefficients}.}
\label{Fig for freq response of all filters} 
\end{figure*}
\begin{table*}[!ht]
\small
\centering
\caption{Coefficients of predict polynomial, update polynomial, and synthesis filters of M-RWTL learned with different sampling rates. $f_s$ is the sampling frequency of the signal in kHz and $N$ is the number of samples of the signal used in experiments.}
\label{Table for Filter Coefficients}
\begin{tabular}{cccc} \hline
\makecell{Sampling rate \\ in two branches} & \makecell{Signal \\ ($f_s,N)$} & \makecell{Predict/Update \\ polynomial} & Filter coefficients \\ \hline
\multirow{2}{*}{$\Big(\frac{2}{3},\frac{1}{3}\Big)$} & \multirow{2}{*}{\makecell{Music-2\\(11.025, 10000)}}   & $t[n]=$[0.4777 0.5101]  & $f_l[n]=$[  0.1707    0.3347    0.3347    0.1599]  \\ 
& & $s[n]=$[0.1751 0.1893] & \makecell{$f_h[n]=$[ -0.0748   -0.1466   -0.1466    0.6867  \\  -0.1586   -0.1586   -0.0758]} \\ \hline

\multirow{2}{*}{$\Big(\frac{1}{2},\frac{1}{2}\Big)$} & \multirow{2}{*}{\makecell{ECG\\(0.36,3600)}}  & $t[n]=$[0.5119 0.5143] & $f_l[n]=$[ 0.2538    0.4935    0.2526]  \\ 
& & $s[n]=$[0.2818 0.2834] & $f_h[n]=$[-0.1449   -0.2818    0.7100   -0.2834   -0.1451]\\ \hline

\multirow{2}{*}{$\Big(\frac{3}{4},\frac{1}{4}\Big)$} & \multirow{2}{*}{\makecell{Music-1\\(11.025, 10000)}}   & $t[n]=$[0.6100    0.6109]  & $f_l[n]=$[  0.1447    0.2369    0.2369    0.2369    0.1445] \\
& & $s[n]=$[0.1496    0.1499] & \makecell{$f_h[n]=$[-0.0755   -0.1236   -0.1236   -0.1236    0.6751 \\  -0.1239   -0.1239   -0.1239   -0.0756]} \\ \hline

\multirow{2}{*}{$\Big(\frac{2}{5},\frac{3}{5}\Big)$} & \multirow{2}{*}{\makecell{Speech\\(11.025,2700)}} & $t[n]=$[0.7048    0.3336] & \makecell{$f_l[n]=$[ 0.0652    0.0652         0    0.0652    0.1955 \\   0.1378    0.1378    0.1955    0.1378]}  \\ 
& & $s[n]=$[ 0.1430    0.4474] & \makecell{$f_h[n]=$[ -0.0137   -0.0137         0   -0.0137         0   -0.0427   -0.0715 \\  -0.0409   -0.0427   -0.0409  -0.0288    0.1959  \\  -0.1280    0.2572   -0.1280    0.1959   0         0   -0.0902]}\\ \hline
\end{tabular}
\end{table*}
\vspace{-1em}
\subsection{Design Examples}
\label{Section for Design Examples}
We present some examples of M-RWTL and the corresponding rational filterbank learned with the proposed method. Table-\ref{Table for design examples} presents parameters used for learning M-RWTL with the following sampling rates in two branches: $\big(\frac{2}{3},\frac{1}{3}\big)$, $\big(\frac{1}{2},\frac{1}{2}\big)$, $\big(\frac{3}{4},\frac{1}{4}\big)$, $\big(\frac{2}{5},\frac{3}{5}\big)$. The first row in table provides values of $q_1$, $q_2$, and $M$. Second to fifth row presents \textit{Lazy} wavelet of the corresponding rational filterbank structure from which the learning is initialized. Sixth and eighth row represents polynomial $R_p(z)$ and $R_u(z)$ respectively. Seventh and ninth row represents the structure of 2-tap predict and update polynomial $T(z)$ and $S(z)$ respectively. 


With these parameters and structure, we learn the M-RWTL matched to given signals of interest. We consider four signals of different types: 1) ECG signal, 2) speech signal, and 3) two music signals (named as music-1 and music-2). These signals are shown in Fig.\ref{Fig for all signals}. We learn M-RWTL with sampling rates $\big(\frac{2}{3},\frac{1}{3}\big)$, $\big(\frac{1}{2},\frac{1}{2}\big)$, $\big(\frac{3}{4},\frac{1}{4}\big)$, $\big(\frac{2}{5},\frac{3}{5}\big)$ in the two branches matched to music-1, music-2, ECG, and speech signals, respectively. Since we learn M-RWTL in the lifting framework that always satisfies PR condition, our learned wavelet system achieves PR with NMSE (normalized mean square error) of the order of $10^{-19}$. Table-\ref{Table for Filter Coefficients} presents coefficients of predict polynomial, update polynomial, and synthesis filters learned with the proposed method. Fig.\ref{Fig for freq response of all filters} shows the frequency response of filters associated with the learned M-RWTL.  

\section{Application in Compressed Sensing based reconstruction}
\label{Section for application}
In this Section, we explore the performance of the learned M-RWTL in the applications of Compressed Sensing (CS)-based reconstruction of 1-D signals. 

Compressed Sensing (CS) problem aims to recover a full signal from a small number of its linear measurement \cite{candes2006robust, donoho2006compressed}. Mathematically, the problem is modeled as:
\begin{equation}
\mathbf{y}=\mathbf{Ax},
\end{equation}  
where $\mathbf{x}$ is the original signal of size $N\times1$, which is compressively measured as $\mathbf{y}$ of size $M \times 1$ by the measurement matrix $\mathbf{A}$ of size $M\times N$.   

Full signal $\mathbf{x}$ is reconstructed from compressive measurements by solving an optimization problem with signal sparsity in some transform domain as prior. Wavelets are extensively used as the sparsifying transforms \cite{donoho2006compressed}. $l_1$ regularized linear least square is solved for signal reconstruction as: 
\begin{equation}
\tilde{\alpha}=\min_{\mathbf{\alpha}}\,\,||\alpha||_1\,\,\,\,\, \text{subject to } \mathbf{y}=\mathbf{AW\alpha},
\end{equation} 
where $\mathbf{x}=\mathbf{W\alpha}$. Here, $\mathbf{W}$ represents the wavelet transform and $\alpha$ is wavelet
transform of $\mathbf{x}$. The above problem is known as basis pursuit (BP) \cite{chen2001atomic} and we used SPGL1 solver to solve the above problem. Full signal is reconstructed as $\tilde{\mathbf{x}}=\mathbf{W\tilde{\alpha}}$.

Table-\ref{Table for CS results} presents CS-based reconstruction performance of M-RWTL with different sampling rates on the four signals shown in Fig.\ref{Fig for all signals}. Reconstruction performance is measured via PSNR (peak signal to noise ratio) given by
\begin{equation}
\text{PSNR}=10\log_{10} \left( \frac{ (max(\mathbf{x}))^2}{MSE} \right ),
\label{Equation for PSNR}
\end{equation}
where $MSE=\frac{1}{N} \sum_{n=0}^{N-1} |x[n]-\tilde{x}[n]|^2$ with $\tilde{\mathbf{x}}$ as the reconstructed signal. Measurement matrix $\mathbf{A}$ is Gaussian and sampling ratio (SR) is varied from 10\% to 90\% with a difference of 10\%, where sampling ratio$=\lfloor \frac{M}{N} \rfloor \times 100$. Three-level wavelet transform decomposition has been used for all the experiments. Results are averaged over 50 independent trials. 

\begin{table*}[]
\centering
\caption{Performance of rational filterbank in CS-based reconstruction of signals. The rational wavelet is learned from one-third of the data samples. Results are averaged over 50 independent trials.}
\label{Table for CS results}
\begin{tabular}{ccccc ccccc cc} \hline
\multirow{2}{*}{Signal} & \multirow{2}{*}{\makecell{Sampling\\ratio (in \%)}} & \multicolumn{9}{c}{PSNR (in dB)}\\ \cline{3-12}
& & db2  & db4  & Bi 5/3 & Bi 9/7 & S$\Big(\frac{2}{3},1\Big)$ & S$\Big(\frac{1}{3},1\Big)$ & S$\Big(\frac{3}{4},1\Big)$ & NG$\Big(\frac{2}{3},\frac{1}{3}\Big)$ & NG$\Big(\frac{1}{3},\frac{2}{3}\Big)$ & NG$\Big(\frac{3}{4},\frac{1}{4}\Big)$ \\ \hline
\multirow{9}{*}{Music-1} & 90 & 28.3 & 31.2 & 28.7 & 33.2 & 29.6 & 24.8 & 29.2 & \textbf{33.3} & 27.6 & 32.7 \\
& 80 & 23.9 & 26.2 & 23.6 & 28.0 & 24.4 & 20.8 & 24.0 & \textbf{28.7} & 23.6 & 27.8 \\
& 70 & 21.0 & 22.5 & 20.1 & 24.1 & 20.6 & 18.2 & 20.6 & \textbf{25.4} & 20.8 & 24.3 \\
& 60 & 18.5 & 19.8 & 17.4 & 20.8 & 17.6 & 16.3 & 17.9 & \textbf{22.6} & 18.6 & 21.3 \\
& 50 & 16.6 & 17.6 & 15.5 & 18.2 & 15.4 & 14.8 & 15.8 & \textbf{19.9} & 16.7 & 18.9 \\
& 40 & 15.1 & 15.8 & 14.0 & 16.1 & 13.8 & 13.7 & 14.2 & \textbf{17.7} & 15.0 & 16.8 \\
& 30 & 13.8 & 14.2 & 12.9 & 14.4 & 12.5 & 12.8 & 12.9 & \textbf{15.6} & 13.6 & 15.1 \\
& 20 & 12.7 & 12.9 & 12.1 & 13.0 & 11.7 & 12.1 & 12.1 & \textbf{13.8} & 12.3 & 13.6 \\
& 10 & 11.9 & 12.0 & 11.5 & 12.1 & 11.4 & 11.6 & 11.6 & \textbf{12.3} & 10.9 & \textbf{12.3} \\ \hline
\multirow{9}{*}{Music-2} &90                 & 29.7 & 30.0 & 29.4   & 30.3   & 29.4        & 29.2        & 29.4       & \textbf{30.6}         & 30.2         & 30.1         \\
& 80                 & 26.1 & 26.4 & 25.9   & 26.8   & 26.0        & 25.7        & 25.9       & \textbf{27.1}         & 26.9         & 26.5         \\
& 70                 & 23.9 & 24.1 & 23.5   & 24.5   & 23.7        & 23.4        & 23.5       & \textbf{24.8}         & \textbf{24.8}         & 24.1         \\
& 60                 & 22.0 & 22.3 & 21.6   & 22.6   & 21.8        & 21.6        & 21.6       & 23.1         & \textbf{23.3}         & 22.2         \\
& 50                 & 20.4 & 20.7 & 19.9   & 21.0   & 20.2        & 20.0        & 19.9       & 21.5         & \textbf{22.0}         & 20.5         \\
& 40                 & 18.9 & 19.1 & 18.4   & 19.4   & 18.6        & 18.5        & 18.2       & 20.0         & \textbf{20.8}         & 18.9         \\
& 30                 & 17.5 & 17.6 & 16.9   & 17.8   & 17.1        & 17.2        & 16.8       & 18.5         & \textbf{19.6}         & 17.5         \\
& 20                 & 16.0 & 16.1 & 15.5   & 16.2   & 15.7        & 15.9        & 15.4       & 16.9         & \textbf{18.3}         & 16.1         \\
& 10                 & 14.7 & 14.7 & 14.2   & 14.8   & 14.4        & 14.7        & 14.3       & 15.3         & \textbf{16.6}         & 14.8         \\ \hline
\multirow{9}{*}{ECG} & 90                 & 49.7 & 50.1 & 49.8   & \textbf{50.9}   & 46.8        & 49.0        & 47.5       & 50.1         & 47.9         & 48.3         \\
& 80                 & 45.9 & 46.4 & 45.9   & \textbf{47.1}   & 42.9        & 45.6        & 42.1       & 45.9         & 44.5         & 43.1         \\
& 70                 & 42.9 & 43.6 & 42.8   & \textbf{44.2}   & 39.6        & 43.3        & 36.2       & 42.6         & 42.4         & 37.2         \\
& 60                 & 39.5 & 40.6 & 39.4   & \textbf{41.2}   & 35.0        & 41.1        & 29.5       & 38.6         & 40.5         & 30.1         \\
& 50                 & 34.6 & 36.1 & 34.1   & 37.3   & 28.2        & \textbf{38.6}        & 24.4       & 33.1         & 38.3         & 25.5         \\
& 40                 & 27.0 & 27.4 & 25.9   & 29.4   & 23.3        & 35.3        & 21.5       & 26.7         & \textbf{35.6}         & 22.7         \\
& 30                 & 21.2 & 21.1 & 20.7   & 21.8   & 20.4        & 30.7        & 19.6       & 22.9         & \textbf{32.1}         & 20.7         \\
& 20                 & 18.5 & 18.3 & 17.7   & 18.7   & 18.5        & 23.8        & 18.2       & 20.3         & \textbf{27.0}         & 18.9         \\
& 10                 & 16.6 & 16.3 & 15.5   & 16.5   & 16.7        & 16.6        & 16.6       & 17.9         & \textbf{20.1}         & 17.0         \\ \hline
\multirow{9}{*}{Speech} & 90                 & 41.7 & 44.7 & 42.2   & \textbf{46.4}   & 43.4        & 38.3        & 43.6       & 43.6         & 39.0         & 42.2         \\
& 80                 & 36.2 & 39.2 & 36.6   & \textbf{41.0}   & 37.6        & 33.5        & 38.0       & 39.1         & 35.1         & 37.1         \\
& 70                 & 31.9 & 34.3 & 32.1   & \textbf{36.4}   & 32.5        & 30.2        & 32.8       & 35.7         & 32.2         & 32.8         \\
& 60                 & 28.5 & 30.1 & 28.1   & 32.1   & 28.4        & 27.4        & 28.4       & \textbf{32.4}         & 29.5         & 29.1         \\
& 50                 & 25.4 & 26.4 & 24.6   & 28.0   & 25.1        & 25.2        & 25.2       & \textbf{29.1}         & 27.2         & 25.9         \\
& 40                 & 22.4 & 23.2 & 21.5   & 24.0   & 22.1        & 22.8        & 22.4       & \textbf{25.7}         & 25.0         & 23.3         \\
& 30                 & 20.1 & 20.6 & 19.2   & 21.2   & 19.9        & 20.6        & 20.2       & \textbf{22.9}         & 22.7         & 20.9         \\
& 20                 & 18.0 & 18.1 & 17.0   & 18.6   & 17.8        & 18.0        & 18.1       & \textbf{20.2}         & 20.2         & 18.8         \\
& 10                 & 16.0 & 15.9 & 15.4   & 16.2   & 15.9        & 15.5        & 16.1       & \textbf{17.6}         & 16.9         & 16.8    \\ \hline    
\end{tabular}
\end{table*}

M-RWTL with the following sampling rate are considered: $\big(\frac{2}{3},\frac{1}{3}\big)$, $\big(\frac{1}{3},\frac{2}{3}\big)$ and $\big(\frac{3}{4},\frac{1}{4}\big)$ and are represented as NG$\big(\frac{2}{3},\frac{1}{3}\big)$, NG$\big(\frac{1}{3},\frac{2}{3}\big)$ and NG$\big(\frac{3}{4},\frac{1}{4}\big)$, respectively. The original signal $\mathbf{x}$ is not available in the compressed sensing application, while the proposed method requires the signal for learning matched rational wavelet. Thus, we propose to sample one-third of the data fully (at 100\% sampling ratio) to learn matched wavelet. Next, we apply the learned matched wavelet for the reconstruction of the rest of the data sampled at lower compressive sensing ratio. One may also use another approach of  \cite{ansari2016joint,ansari2017image} to learn matched wavelet in CS application. However, so far \cite{ansari2016joint} and \cite{ansari2017image} are limited to the special case of dyadic matched wavelet and can be explored for extension in the CS application with rational wavelet transform learning as a future work.
 
The reconstruction performance is compared with standard orthogonal Daubechies wavelets db2 and db4, and standard bi-orthogonal wavelets, bior5/3 and bior9/7 (labeled as Bi 5/3 and Bi 9/7, respectively). The performance is also compared with overcomplete rational wavelets designed in \cite{bayram2009frequency}. For fair comparison, we consider the same sampling rate in the low frequency branch for these overcomplete rational wavelets as used in proposed M-RWTL in Table-III, i.e., we use overcomplete rational wavelets \cite{bayram2009frequency} with sampling rates: $\big(\frac{2}{3},1\big)$, $\big(\frac{1}{3},1\big)$ and $\big(\frac{3}{4},1\big)$, represented as S$\big(\frac{2}{3},1\big)$, S$\big(\frac{1}{3},1\big)$ and S$\big(\frac{3}{4},1\big)$. 

From Table \ref{Table for CS results}, it is observed that Bi 9/7 performs best among the existing wavelets used. Also, the overcomplete rational wavelet $\big(\frac{1}{3},1\big)$ performs better than the existing wavelets on ECG signal at all sampling ratios less than 60\%, while these perform comparable or inferior in performance to the  existing wavelets on music and speech signals. 

On the other hand, the proposed M-RWTL perform better, with an improvement of upto 1.8 dB of PSNR, in comparison to existing wavelets for music signals. Particularly, M-RWTL with $\big(\frac{2}{3},\frac{1}{3}\big)$ performs better on music-1 signal for all the sampling ratios. M-RWTL with $\big(\frac{2}{3},\frac{1}{3}\big)$ performs better on music-2 signal from 90\% to 70\% sampling ratio beyond which $\big(\frac{1}{3},\frac{2}{3}\big)$ performs better. On ECG signal, existing wavelet Bi 9/7 performs better than all other wavelets at higher sampling ratios from 90\% to 60\%. At 50\% sampling ratio, overcomplete rational wavelet $\big(\frac{1}{3},1\big)$ outperforms all existing and proposed rational wavelets. Below 50\% sampling ratio, M-RWTL with $\big(\frac{1}{3},\frac{2}{3}\big)$ performs better than all existing as well as overcomplete rational wavelets with an improvement of upto 10 dB than existing wavelets and upto 4 dB than overcomplete rational wavelets. Similarly, in case of speech signal, Bi 9/7 performs better than all other wavelets from sampling ratio 90\% to 70\%. Below 70\% sampling ratio, M-RWTL with $\big(\frac{2}{3},\frac{1}{3}\big)$ outperforms all existing as well as overcomplete rational with an improvement of upto 1.7 dB. 

Further, this is to note that at higher sampling ratios, PSNR of the reconstructed ECG and speech signals is high at around 40 dB and 30 dB, respectively, with different wavelets. Hence, the reconstructed signal appears almost similar to the original signal. The quality of the reconstructed signal deteriorates with decreasing sampling ratios, where the proposed M-RWTL performs best with as much as 10 dB improvement. Overall, the performance of rational matched wavelets is superior at lower sampling ratios in compressive sensing application. In this paper, we have not explored the problem of choosing the optimal sampling rate of learned wavelet for a particular signal or in a particular application. This remains an open problem and can be explored in the future.


\section{Conclusion}
Theory of learning rational wavelet transform in the lifting framework, namely, the M-RWTL method, has been presented. The existing theory of lifting framework is extended from dyadic to rational wavelets, where critically sampled rational matched wavelet filterbank can be designed for any general rational sampling ratios. The concept of rate converters is introduced to handle variable data rate of subbands. The learned signal-matched rational filterbank inherits all the advantages of lifting framework. The learned analysis and synthesis filters are FIR and are easily implementable in hardware, thus making RWT easily usable in applications. Closed form solution is presented for learning rational wavelet and thus, no greedy solution is required making M-RWTL computationally efficient. 

The proposed M-RWTL transform can be learned from a short snapshot of a single signal and hence, extends the use of transform learning from the requirement of large training data to small data snapshots. As a proof of concept, the learned M-RWTL is applied in CS-based reconstruction of signals and is observed to perform better compared to the existing wavelets. Although the learned M-RWTL performs better in the above application, it is not known apriori as to what sampling ratios in the two branches of the rational filterbank are optimal for a given signal in a particular application. We leave it as an open problem for the future work. 

\ifCLASSOPTIONcaptionsoff
  \newpage
\fi
\bibliographystyle{IEEEtran}
\bibliography{refs}

\end{document}